\documentclass[dvipsnames]{lmcs}
\pdfoutput=1

\usepackage[utf8]{inputenc}

\usepackage{lastpage}
\lmcsdoi{20}{3}{12}
\lmcsheading{}{\pageref{LastPage}}{}{}%
{Aug.~31,~2022}{Aug.~05,~2024}{}

\usepackage{amsfonts, amsmath, amssymb,  latexsym}
\usepackage [english]{babel}
\usepackage{graphicx}
\usepackage{listings}
\usepackage{hyperref}
\usepackage{tikz}
\usepackage{xcolor}
\usepackage{enumitem}
\usepackage{tikzpagenodes}
\usetikzlibrary{automata, positioning, arrows}
\tikzset{
	->, 
	node distance=2cm, 
	initial text=$ $, 
}

\theoremstyle{plain}\newtheorem*{thmA}{Theorem A}
\theoremstyle{plain}\newtheorem*{thmB}{Theorem B}
\theoremstyle{plain}\newtheorem*{thmD}{Theorem C}

\usepackage [autostyle, english = american]{csquotes}
\MakeOuterQuote{"}
\usepackage{harpoon}

\newcommand{\cal}[1]{\mathcal{#1}}

\def\modd#1#2{#1\ \mbox{\rm (mod}\ #2\mbox{\rm )}}
\def\Z{\mathbb{Z}}
\def\R{\mathbb{R}}
\def\N{\mathbb{N}}
\def\Q{\mathbb{Q}}

\theoremstyle{defC}
\newtheorem{factC}[thm]{Fact}

\def\SigmaH{\Sigma_{\#}}
\def \Th {{\rm {FO}}}
\def \Irr {{\rm {\textbf{Irr}}}}
\def\Qu{\Irr_{quad}}

\lstdefinelanguage{pecan}{
	keywords=[1]{forall, exists, max, min, sup, inf, are, is, if, then, match, with, case, end, let, be, in, else, iff},
	keywordstyle=[1]\color{blue}\bfseries,
	keywords=[2]{false, true, sometimes},
	commentstyle=\color{CadetBlue}\textit,
	stringstyle=\color{ForestGreen}, 
	keywordstyle=[2]\color{orange}\bfseries,
	keywords=[3]{assert_prop,Structure,defining,Theorem,Prove,Example,Alias,Restrict,Define,Display,Execute,load,shuffle,import,save_aut,save_aut_img,that,context,end_context,forget,shuffle,shuffle_or,using,of},
	keywordstyle=[3]\color{teal}\bfseries,
	keywords=[4]{@annotation,@postprocess,@no_simplify,@simplify,@simplify_states,@simplify_edges},
	keywordstyle=[4]\color{purple}\bfseries,
	literate=*%
	    {\#}{{{\color{teal}\bfseries\#}}}1
	    {+}{{{\color{red}+~}}}1
	    {*}{{{\color{red}$\cdot$~}}}1
	    {/}{{{\color{red}/~}}}1
	    {-}{{{\color{red}-~}}}1
        {:=}{{{\color{red}:=~}}}1
        {..}{{{\color{red}..~}}}1
        {\{}{{{\color{red}\{}}}1
        {\}}{{{\color{red}\}}}}1
        {|}{{{$\color{red} \lor~$}}}1
        {*}{{{\color{red}*~}}}1
        {:}{{{\color{red}:~}}}1
        {>}{{{\color{red}>~}}}1
        {<}{{{\color{red}<~}}}1
        {<=>}{{{$\color{red}\Leftrightarrow~$}}}1
        {.}{{{\color{red}.~}}}1
        {&}{{{$\color{red} \land~$}}}1
        {!}{{{$\color{red}\lnot~$}}}1
        {!=}{{{$\color{red} \neq$}}}1
        {=}{{{\color{red}=~}}}1
        {exists }{{{$\color{red}\exists$}}}1
        {forall }{{{$\color{red}\forall$}}}1,
    sensitive=false, 
    morecomment=[l]{//}, 
    morecomment=[s]{/*}{*/}, 
    morestring=[b]", 
    showstringspaces=false
}

\lstnewenvironment{pecan}
  {
    \lstset{
        language=pecan, 
        basicstyle=\small\ttfamily, 
        mathescape=true
        }
  }
  {
  }

\lstnewenvironment{pecan_output}
  {
    \lstset{
        basicstyle=\small\ttfamily,
        mathescape=true
        }
  }
  {
  }

\newcommand{\pecaninline}[1]{\lstinline[language=pecan,basicstyle=\small\ttfamily,mathescape=true]{#1}}

\begin{document}

\title[Decidability for Sturmian words]{Decidability for Sturmian words}
\author[P. Hieronymi]{Philipp Hieronymi}[a]
\address
{Mathematisches Institut\\
Universität Bonn\\
Endenicher Allee 60\\
D-53115 Bonn\\
Germany}
\email{hieronymi@math.uni-bonn.de}
\urladdr{https://www.math.uni-bonn.de/people/phierony/}

\author[D. Ma]{Dun Ma}[b]
\address{Department of Computer Science and Engineering\\
University of California, San Diego\\
9500 Gilman Drive\\
La Jolla, CA 92093-0404\\
USA
}
\email{d4ma@ucsd.edu}

\author[R. Oei]{Reed Oei}[c]
\address{Department of Mathematics\\University of Illinois at Urbana-Champaign\\1409 West Green Street\\Urbana, IL 61801\\ USA}

\author[L. Schaeffer]{Luke Schaeffer}[d]
\address{
Institute for Quantum Computing\\
University of Waterloo\\
Waterloo, Ontario\\
N2L 3G1\\
Canada
}
\email{lrschaeffer@gmail.com}

\author[C. Schulz]{Chris Schulz}[e,f]
\address{Department of Mathematics\\University of Illinois at Urbana-Champaign\\1409 West Green Street\\Urbana, IL 61801\\ USA}
\address{Department of Pure Mathematics\\
200 University Avenue West\\
Waterloo, Ontario\\
N2L 3G1\\
Canada
}
\email{chris.schulz@uwaterloo.ca}

\author[J. Shallit]{Jeffrey Shallit}[g]
\address{School of Computer Science\\
University of Waterloo\\ 
Waterloo, Ontario\\
N2L 3G1\\Canada
}
\email{shallit@uwaterloo.ca}
\urladdr{https://cs.uwaterloo.ca/\textasciitilde shallit/}

\dedicatory{In Memory of Reed Oei (1999-2022)} 

 \maketitle

\begin{abstract} We show that the first-order theory of Sturmian words over Presburger arithmetic is decidable.
Using a general adder recognizing addition in Ostrowski numeration systems by Baranwal, Schaeffer and Shallit, we prove that the first-order expansions of Presburger arithmetic by a single Sturmian word are uniformly $\omega$-automatic, and then deduce the decidability of the theory of the class of such structures.
Using an implementation of this decision algorithm called Pecan, we automatically reprove classical theorems about Sturmian words in seconds, and are able to obtain new results about antisquares and antipalindromes in characteristic Sturmian words. 
\end{abstract}
\section{Introduction}

It has been known for some time that, for
certain infinite words ${\bf c} = c_0 c_1 c_2 \cdots$ over
a finite alphabet $\Sigma$, the
first-order logical theory ${\rm FO}(\N, <, +, 0, 1, n \mapsto c_n)$
is decidable.   In the case where ${\bf c}$ is a
{\bf $k$-automatic sequence} for $k \geq 2$, this is due
to B\"uchi \cite{Buechi}, although his original proof was flawed.
The correct statement appears, for example, in Bruy\`ere et al.~\cite{BHMV,BHMV-Correction}.
Although the worst-case running time of the decision procedure
is truly
formidable (and non-elementary), it turns out that an implementation
can, in many cases, decide the truth of interesting and
nontrivial first-order statements
about automatic sequences in a reasonable length of time.
Thus, one can easily reprove known results, and obtain new ones,
merely by translating the desired result into the appropriate first-order
statement $\varphi$ and running the decision procedure on $\varphi$.
For an example of the kinds of things that can be proved, see Go\v{c}, Henshall, and Shallit
\cite{GHS}.\newline

\noindent More generally, the same ideas can be used for other kinds of sequences
defined in terms of some numeration system for the natural numbers.
Such a numeration system provides a unique (up to leading zeros) representation
for $n$ as a sum of terms of some other sequence $(s_n)_{n \geq 1}$.
If the sequence ${\bf c} = c_0 c_1 c_2 \cdots$
can be computed by a finite automaton taking the
representation of $n$ as input, and if further, the addition of represented
integers is computable by another finite automaton, then
once again the first-order theory ${\rm FO}(\N, <, +, 0, 1, n \mapsto c_n)$ is decidable.  This is the case, for example, for the
so-called Fibonacci-automatic sequences in Mousavi, Schaeffer, and Shallit
\cite{MSS1}
and the Pell-automatic sequences in Baranwal and Shallit
\cite{BS1}.\newline

\noindent More generally, the same kinds of ideas can handle Sturmian words.
For quadratic numbers, this was first observed by
Hieronymi and Terry \cite{HT}.   In this
paper we extend those results to all Sturmian characteristic words.
Thus, the first-order theory of Sturmian characteristic words is
decidable.  As a result, many classical theorems about Sturmian
words, which previously required intricate proofs, can be proved
automatically by a theorem-prover in a few seconds.  As examples,
in Section 7 we reprove basic results such as the balanced property and
the subword complexity of these words.\newline

\noindent 
Let $\alpha,\rho\in \R$ be such that $\alpha$ is irrational. The \textbf{Sturmian word with slope $\alpha$ and intercept $\rho$} is the infinite $\{0,1\}$-word  ${\bf c_{\alpha,\rho}} = c_{\alpha,\rho}(1)c_{\alpha,\rho}(2)\cdots $ such that for all $n\in \N$
\[
c_{\alpha,\rho}(n) = \lfloor \alpha (n+1) + \rho \rfloor - \lfloor \alpha n + \rho \rfloor - \lfloor \alpha \rfloor.
\]
When $\rho =0$, we call ${\bf c_{\alpha,0}}$ the \textbf{characteristic word} of slope $\alpha$. Sturmian words and their combinatorical properties have been studied extensively. We refer the reader to the survey by Berstel and Séébold \cite[Chapter 2]{zbMATH01737190}. Note that ${\bf c_{\alpha,\rho}}$ can be understood as a function from $\N$ to $\{0,1\}$. Let $\mathcal{L}$ be the signature\footnote{In model theory this is usually called (or identified with) the language of the theory. However, here this conflicts with the convention of calling an arbitrary set of words a language.} of the first-order logical theory ${\rm FO}(\N,<,+,0,1)$ and denote by $\mathcal{L}_c$ the signature obtained by adding a single unary function symbol $c$ to $\mathcal{L}$. Now let $\cal N_{\alpha,\rho}$ be the $\mathcal{L}_c$-structure $(\N,<,+,0,1,n \mapsto c_{\alpha,\rho}(n))$, where  we expand Presburger arithmetic by a Sturmian word interpreted as a unary function. The main result of this paper is the decidability of the theory of the collection of such expansions. Set $\Irr := (0,1)\setminus \Q$. Let $\mathcal{K}_{\rm sturmian} := \{ \mathcal{N}_{\alpha,\rho} \ : \ \alpha \in \Irr, \rho \in \R\}$, and let $\mathcal{K}_{\rm char} := \{ \mathcal{N}_{\alpha,0} \ : \ \alpha \in \Irr\}$.
\begin{thmA}
The first-order logical theories\footnote{Given a signature $\mathcal{L}_0$ and a class $\cal K$ of $\mathcal{L}_0$-structures, the first-order logical theory  of $\cal K$ is defined as the set of all $\mathcal{L}_0$-sentences that are true in all structures in $\cal K$. This theory is denoted by ${\rm FO}(\cal{K})$.} ${\rm FO}(\mathcal{K}_{\rm sturmian})$ and ${\rm FO}(\mathcal{K}_{\rm char})$ are decidable.
\end{thmA}

\noindent So far, decidability was only known for individual ${\rm FO}(\mathcal N_{\alpha,\rho})$, and only for very particular $\alpha$. By \cite{HT} the logical theory ${\rm FO}(\mathcal N_{\alpha,0})$ is decidable when $\alpha$ is a quadratic irrational\footnote{A real number is \textbf{quadratic} if it is the root of a quadratic equation with integer coefficients.}. Moreover, if the continued fraction of $\alpha$ is not computable, it can be seen rather easily that  ${\rm FO}(\mathcal N_{\alpha,0})$ is undecidable. \newline		

\noindent Theorem A is rather powerful, as it allows to automatically decide combinatorial statements about all Sturmian words. Consider the $\mathcal{L}_c$-sentence $\varphi$
\[
\forall p \ (p>0) \rightarrow \Big(\forall i \ \exists j \ j > i \wedge c(j) \neq c(j+p)\Big). 
\]
We observe that $\mathcal N_{\alpha,\rho} \models \varphi$ if and only if ${\bf c_{\alpha,\rho}}$ is not eventually periodic. Thus the decision procedure from Theorem A allows us to check that no Sturmian word is eventually periodic. Of course, it is well-known that no Sturmian word is eventually periodic, but this example indicates potential applications of Theorem A. We outline some of these in Section \ref{section:pecan}.\newline

\noindent We not only prove Theorem A, but instead establish a vastly more general theorem of which Theorem A is an immediate corollary. To state this general result, let $\mathcal L_m$ be the signature of ${\rm FO}(\R,<,+,\Z)$; that is, the signature of 
${\rm FO}(\R,<,+)$ together with a unary predicate for $\Z$. Let $\mathcal L_{m,a}$ be the extension of $\mathcal L_m$ by another unary predicate. For $\alpha \in \R_{>0}$, we let $\mathcal{R}_\alpha$ denote $\mathcal L_{m,a}$-structure $(\mathbb{R}, <, +, \mathbb{Z}, \alpha \mathbb{Z})$. When $\alpha \in \Q$, it has long been known that ${\rm FO}(\mathcal R_{\alpha})$ is decidable (arguably due to Skolem \cite{skolem}). Recently this result was extended to quadratic numbers.

\begin{fact}[{Hieronymi \cite[Theorem A]{H}}] \label{fact:quad} Let $\alpha$ be a quadratic irrational. Then ${\rm FO}(\cal R_{\alpha})$ is decidable.
\end{fact}

\noindent See also Hieronymi, Nguyen and Pak \cite{HNP} for a computational complexity analysis of this decision procedure. The proof of Fact 1.1 establishes that if $\alpha$ is quadratic, then $\cal R_{\alpha}$ is an $\omega$-automatic structure; that is, it can be represented by B\"uchi automata. Since every $\omega$-automatic structure has a decidable first-order theory, so does $\cal R_{\alpha}$. See Khoussainov and Minnes \cite{KMinnes} for a survey on $\omega$-automatic structures.  
The key insight needed to prove $\omega$-automaticity of $\cal R_{\alpha}$ is that addition in the Ostrowski-numeration system based on $\alpha$ is recognizable by a B\"uchi automaton when $\alpha$ is quadratic. See Section \ref{section:prelim} for a definition of Ostrowski numeration systems.\newline

\noindent As observed in \cite{H}, there are  examples of non-quadratic irrationals $\alpha$ such that $\cal R_{\alpha}$ has an undecidable theory and hence is not $\omega$-automatic. However, in this paper we show that the common theory of the $\cal R_{\alpha}$ is decidable. Let $\mathcal{K}$ denote the class of $\mathcal{L}_{m,a}$-structures $
\{ \mathcal R_{\alpha} \ : \ \alpha \in \Irr\}$.

\begin{thmB}
The theory ${\rm FO}(\mathcal{K})$ is decidable.
\end{thmB}

\noindent Indeed, we will even prove a substantial generalization of Theorem B. For each $\mathcal L_{m,a}$-sentence $\varphi$, we set
$
M_{\varphi} := \{ \alpha \in \Irr \ : \ \mathcal{R}_{\alpha}\models \varphi \}$.
Let $\Qu$ be the set of all quadratic irrational real numbers in $\Irr$.
Define 
\[
\mathcal{M}:=(\Irr,<,(M_{\varphi})_{\varphi},(\mathnormal{q})_{\mathnormal{q}\in \Qu})
\]
to be the expansion of the dense linear order $(\Irr,<)$ by
 predicates for $M_{\varphi}$ for each $\mathcal{L}_{m,a}$-sentence $\varphi$, and constant symbols for each quadratic irrational real number in $\Irr$.

\begin{thmD}
The theory ${\rm FO}(\cal M)$ is decidable.
\end{thmD}

\noindent Observe that Fact 1.1 and Theorem B follow immediately from Theorem C. 
\noindent We outline how Theorem B implies Theorem A. Note that for every irrational $\alpha$, the structure $\mathcal R_{\alpha}$ defines the usual floor function $\lfloor \cdot \rfloor : \R \to \Z$, the singleton $\{\alpha\}$ and the successor function on $\alpha\Z$. Hence $\mathcal R_{\alpha}$ also defines the set $\{ (\rho,\alpha n,c_{\alpha,\rho}(n)) \ : \ \rho \in \R, n \in \N\}$. From the definability of $\{\alpha\}$, we have that the function from $\alpha \N$ to $\{0,\alpha\}$ given by $\alpha n \mapsto \alpha c_{\alpha,\rho}(n)$ is definable in $\mathcal{R}_{\alpha}$. Thus the $\mathcal{L}_c$-structure $(\alpha \N,<,+,0,\alpha,\alpha n \mapsto \alpha c_{\alpha,\rho}(n))$ can be defined in $\mathcal{R}_{\alpha}$, and this definition is uniform in $\alpha$. Since the former structure is $\mathcal{L}_c$-isomorphic to $\mathcal{N}_{\alpha,\rho}$, we have that for every $\mathcal{L}_c$-sentence $\varphi$ there is an $\mathcal{L}_{m,a}$-formula $\psi(x)$ such that
\begin{itemize}
    \item $
\varphi \in \Th(\mathcal{K}_{\rm sturmian}) \hbox{ if and only if } \forall x \ \psi(x) \in \Th(\mathcal{K})
$ and
   \item $
\varphi \in \Th(\mathcal{K}_{\rm char}) \hbox{ if and only if } \psi(0) \in \Th(\mathcal{K}).
$
\end{itemize}
Even Theorem C is not the most general result we prove. Its statement is more technical and we postpone it until Section \ref{section:decidability}. However, we
want to point out that we can add predicates for interesting subsets of $\Irr$ to $\cal M$ without changing the decidability of the theory. Examples of such subsets are the set of all $\alpha\in \Irr$ such that the terms in the continued fraction expansion of $\alpha$ are powers of 2, or the set of all $\alpha\in \Irr$ such that the terms in the continued fraction expansion of $\alpha$ are not in some fixed finite set. This means we can not only automatically prove theorems about all characteristic Sturmian words, but also prove theorems about all characteristic Sturmian words whose slope is one of these sets. There is a limit to this technique. 
 If we add a predicate for the set of all $\alpha\in \Irr$ such that the terms of continued fraction expansion of $\alpha$ are bounded, or add a predicate for the set of elements in $\Irr$ whose continued fractions has strictly increasing terms, then our method is unable to conclude whether the resulting structure has a decidable theory. See Section \ref{section:decidability} for a more precise statement about what kind of predicates can be added.\newline

\noindent The proof of Theorem C follows closely the proof from \cite{H} of the $\omega$-automaticity of $\cal R_{\alpha}$ for fixed quadratic $\alpha$. Here we show that the construction of the B\"uchi automata needed to represent $\cal R_{\alpha}$ is actually uniform in $\alpha$. See Abu Zaid, Grädel and Reinhardt \cite{AZGR} for a systematic study of uniformly automatic classes of structures. Deducing Theorem C from this result is then rather straightforward. The key ingredient to establish the $\omega$-automaticity of  $\cal R_{\alpha}$ is an automaton that can perform addition in Ostrowski-numeration systems. By \cite{HT} there is an automaton that recognizes the addition relation for $\alpha$-Ostrowski numeration systems for fixed quadratic $\alpha$. So for a fixed quadratic number, there exists a $3$-input automaton that accepts the $\alpha$-Ostrowski representations of all triples of natural numbers $x,y,z$ with $x+y=z$. In order to prove Theorem C, we need a uniform version of such an adder. This general adder is described in Baranwal, Schaeffer, and Shallit \cite{BSS}. There a 4-input automaton is constructed that accepts 4-tuples consisting of an encoding of a real number $\alpha$ and three $\alpha$-Ostrowski representations of natural numbers $x,y,z$ with $x+y=z$. See Section \ref{section:addition} for details. \newline

\noindent As mentioned above, an implementation of the decision algorithm provided by Theorem A can be used to study Sturmian words. We created a software program called Pecan \cite{pecan-repo} that includes such an implementation. Pecan is inspired by Walnut~\cite{walnut} by Mousavi, an automated theorem-prover for deciding properties of automatic words. The main difference is that Walnut is based on finite automata, while Pecan uses B\"uchi automata. In our setting it is more convenient to work with B\"uchi automata instead of finite automata, since the infinite families of words we want to consider---like Sturmian words---are indexed by real numbers. Section \ref{section:pecan} provides more information about Pecan and contains further examples how Pecan is used to prove statements about Sturmian words. Pecan's implementation is discussed in more detail in \cite{Pecan}.\newline

\noindent This is an extended version of the paper \cite{SturmianCSL} presented at CSL 2022.


\subsection*{Acknowledgments} Part of this work was done in the research project ``Building a theorem-prover'' at the Illinois Geometry Lab in Spring 2020. P.H. and C.S. were partially supported by NSF grant DMS-1654725. P.H. was partially supported by the Hausdorff Center for Mathematics at the University of Bonn. We thank Mary Angelica Gramcko-Tursi and Sven Manthe for carefully reading a draft of this paper.

\section{Preliminaries}
\label{section:prelim}

Throughout, $i,j,k,\ell,m,n$ are used for natural numbers.
Let $X,Y$ be two sets and $Z\subseteq X\times Y$. For $x\in X$, we let $Z_x$ denote the set $\{y \in Y \ : \ (x,y) \in Z\}.$ Similarly, given a function $f : X\times Y \to W$ and $x\in X$, we write $f_x$ for the function $f_x : Y \to W$ that maps $y\in Y$ to $f(x,y)$.\newline

\noindent Given a (possibly infinite word) $w$ over an alphabet $\Sigma$, we write $w_i$ for the $i$-th letter of $w$, and $w|_n$ for $w_1\cdots w_n$. We write $|w|$ for the length of $w$. We let $\Sigma^{\omega}$ denote the set of infinite words over $\Sigma$. If $\Sigma$ is totally ordered by $\prec$, we let $\prec_{\rm lex}$ denote the corresponding lexicographic order on $\Sigma^{\omega}$. Letting $u, v \in \Sigma^\omega$, we also write $u \prec_{\rm colex} v$ if there is a maximal $i$ such that $u_i \neq v_i$, and $u_i < v_i$ for this $i$. Note that while $\prec_{\rm lex}$ is a total order on $\Sigma^\omega$, the order $\prec_{\rm colex}$ is only a partial order. However, for a given $\sigma \in \Sigma$, the order $\prec_{\rm colex}$ is a total order on the set of all words $v \in \Sigma^\omega$ such that $v_j$ is eventually equal to $\sigma$.\newline

\noindent We will also need to apply $\prec_{\rm lex}$ and $\prec_{\rm colex}$ to finite sequences $u, v$ of the same length. We do this by choosing a $\sigma \in \Sigma$ (the choice does not matter) and stating that $u \prec_{\rm lex} v$ iff $u\sigma^\omega \prec_{\rm lex} v\sigma^\omega$, and similarly for $\prec_{\rm colex}$.\newline

\noindent A \textbf{B\"uchi automaton (over an alphabet $\Sigma$)} is a quintuple $\mathcal{A} = (Q,\Sigma,\Delta,I,F)$ where $Q$ is a finite set of states, $\Sigma$ is a finite alphabet, $\Delta\subseteq Q \times \Sigma \times Q$ is a transition relation, $I\subseteq Q$ is a set of initial states, and $F \subseteq Q$ is a set of accept states. \newline

\noindent Let $\mathcal{A} = (Q,\Sigma,\Delta,I,F)$ be a B\"uchi automaton. Let $\sigma\in \Sigma^{\omega}$. A \textbf{run of $\sigma$ from $p$} is an infinite sequence $s$ of states in $Q$ such that $s_0 = p$, $(s_n,\sigma_n,s_{n+1})\in \Delta$ for all $n<|\sigma|$. If $p \in I$, we say $s$ is a \textbf{run of $\sigma$}. Then $\sigma$ is \textbf{accepted by $\cal A$} if there is a run $s_0s_1\cdots$ of $\sigma$ such that $\{n:s_n \in F\}$ is infinite. We call this run an accepting run. We let $L(\cal A)$ be the set of words accepted by $\cal A$. \newline

\noindent If for every state $s$ in $\mathcal{A}$ there is a run of some string from an initial state through $s$ to an accept state, where $s$ is not the last state in the run, then we say $\mathcal{A}$ is \textbf{trim.} Every B\"uchi automaton has an equivalent trim automaton, which may be obtained simply by removing (possibly iteratively) every state failing this condition. There are other types of $\omega$-automata with different acceptance conditions, but in this paper we only consider B\"uchi automata.\newline 

\noindent Let $\Sigma$ be a finite alphabet. We say a subset $X\subseteq \Sigma^{\omega}$ is \textbf{$\omega$-regular} if it is recognized by some B\"uchi automaton. Let $u_1,\dots,u_n \in \Sigma^{\omega}$. We define the \textbf{convolution}  $c(u_1,\dots,u_n)$ of $u_1,\dots, u_n$ as the element of $(\Sigma^n)^{\omega}$ whose value at position $i$ is the $n$-tuple consisting of the values of $u_1,\dots, u_n$ at position $i$. We say that $X \subseteq (\Sigma^{\omega})^n$ is \textbf{$\omega$-regular} if $c(X)$ is $\omega$-regular.

\begin{fact}\label{fact:omegaclosure}
The collection of $\omega$-regular sets is closed under union, intersection, complementation and projection.
\end{fact}

\noindent Closure under complementation is due to B\"uchi \cite{Buechi}. We refer the reader to Khoussainov and Nerode \cite{aut_theory} for more information and a proof of Fact \ref{fact:omegaclosure}. As consequence of Fact \ref{fact:omegaclosure}, we have that for every $\omega$-regular subset $W\subseteq (\Sigma^{\omega})^{m+n}$ the set
\[
\{ s \in (\Sigma^{\omega})^{m} \ : \ \forall t \in (\Sigma^{\omega})^{n} \ (s,t) \in W\}  
\]
is also $\omega$-regular. \newline

\noindent The proof of Theorem \ref{oplus_fin_regular} will utilize a few other related types of automaton. A \textbf{finite automaton} has the same internal structure as a B\"uchi automaton i.e. is also a quintuple $\mathcal{A} = (Q, \Sigma, \Delta, I, F)$ with the same restrictions, but it takes a finite word $\sigma \in \Sigma^*$ as input. In the case of a finite automaton, runs are finite sequences instead of infinite sequences but otherwise follow the same rule on transitions. We say that \textbf{$\sigma$ is accepted by $\mathcal{A}$} in this case if there is a run of $\sigma$ such that $s_{|\sigma|} \in F$. \newline

\noindent We will also refer to \textbf{general} finite and B\"uchi automata. These are the same as finite and B\"uchi automata, respectively, but where $\Sigma$ is no longer required to be a finite alphabet. Note that $Q$ is still finite in these cases; therefore $\Delta$, viewed as a directed multigraph on $Q$, still has finitely many vertices but may have infinitely many arrows between the same pair of vertices. General finite and B\"uchi automata are not often considered, as they do not have the same computability properties\footnote{To see why, consider e.g. a generalized B\"uchi automaton recognizing words over $\N$ consisting of a single initial state $q_0$ and a single final state $q_1$ such that there is a noncomputable set $S\subseteq \N$ with
$
\Delta = \{ (q_0,s,q_1) \ : \ s \in S\}.$}
, but they may sometimes be converted into ``equivalent'' finite and B\"uchi automata, as we will see in Section \ref{section:addition}.

\subsection{\texorpdfstring{$\omega$}{Omega}-regular structures}

\noindent Let $\mathcal{U} =(U;R_1,\dots,R_m)$ be a structure, where $U$ is a non-empty set and $R_1,\dots, R_m$ are relations on $U$. We say $\cal U$ is \textbf{$\omega$-regular} if its domain and its relations are $\omega$-regular.\newline

\noindent B\"uchi's theorem \cite{Buechi} on the decidability of the monadic second-order theory of one successor immediately gives the following well-known fact.

\begin{fact} Let $\mathcal{U}$ be an $\omega$-regular structure. Then the theory ${\rm FO}(\mathcal{U})$ is decidable.
\end{fact}

\noindent In this paper, we will consider families of $\omega$-regular structures that are uniform in the following sense.
Fix $m \in \N$ and a map $\operatorname{ar} : \{1,\dots ,m\} \to \N$. 
Let $Z$ be a set and for $z\in Z$ let $\mathcal{U}_z$ be a structure $(U_z;R_{1,z},\dots,R_{m,z})$ such that $R_{i,z} \subseteq U_z^{ar(i)}$.
We say that $(\mathcal{U}_z)_{z\in Z}$ is a \textbf{uniform family of $\omega$-regular structures} if
\begin{itemize}
    \item $\{ (z,y) \ : \ y \in U_z\}$ is $\omega$-regular,
    \item $\{ (z,y_1,\dots,y_{ar(i)}) \ : \ (y_1,\dots,y_{ar(i)}) \in R_{i,z}\}$ is $\omega$-regular for each $i\in\{1,\dots, m\}$.
\end{itemize}
We refer the reader to \cite{AZGR} for an in-depth analysis of uniformity in automatic structure.\newline

\noindent From B\"uchi's theorem, we immediately obtain the following.

\begin{fact}\label{fact:uniformregulardef}
Let $(\mathcal{U}_z)_{z\in Z}$ be a uniform family of $\omega$-regular structures, and let $\varphi$ be a formula in the signature of these structures. Then the set 
\[
\{ (z,u) \ : \ z\in Z, u \in U_z, \ \mathcal{U}_z \models \varphi( u )\}
\]
is $\omega$-regular, and, the automaton recognizing this set can be effectively computed given $\varphi$. Moreover, the theory ${\rm FO}(
\{ \mathcal{U}_z \ : \ z \in Z\})
$
is decidable.
\end{fact}
\begin{proof} 
When $\varphi$ is an atomic formula, the statement follows immediately from the definition of a uniform family of $\omega$-regular structures and the $\omega$-regularity of equality. By Fact \ref{fact:omegaclosure}, the statement holds for all formulas. 
\end{proof}

\noindent Let $w \in \Sigma^{\omega}$. The \textbf{acceptance problem for $w$} is the following decision problem:
\[
\text{Given a B\"uchi automaton $\mathcal{A}$ over $\Sigma$, is $w$ accepted by $\mathcal{A}$?}
\]
For examples of non-$\omega$-regular words with a decidable acceptance problem, see Elgot and Rabin \cite{ER}, Semenov \cite{Semenov83} or Carton and Thomas \cite{CartonThomas}. We obtain the following well-known corollary of Fact \ref{fact:uniformregulardef}.

\begin{fact}\label{fact:parameter}
Let $(\mathcal{U}_z)_{z\in Z}$ be a uniform family of $\omega$-regular structures, and let $w \in Z$ be such that the acceptance problem for $w$ is decidable. Then the theory
${\rm FO}(\mathcal{U}_w)
$
is decidable.
\end{fact}

\subsection{Binary representations}

For $k\in \N_{>1}$ and $b=b_0b_{1}b_2\cdots b_n \in \{0,1,\dots,k-1\}^{*}$, we define $[b]_k := \sum_{i=0}^n b_i k^i$. For $N\in \N$ we say $b\in \{0,1\}^{*}$ is a \textbf{binary representation} of $N$ if $[b]_2 = N$.\newline

\noindent Throughout this paper, we will often consider infinite words over the (infinite) alphabet $\{0,1\}^{*}$. Let $[\cdot]_2 : (\{0,1\}^{*})^{\omega} \to \N^{\omega}$ be the function that maps $u=u_1u_2\cdots \in (\{0,1\}^*)^{\omega}$ to
\[
[u_1]_2[u_2]_2[u_3]_2\cdots.
\]
We will consider the following  different relations on $(\{0,1\}^{*})^{\omega}$. \newline

\noindent Let $u,v \in (\{0,1\}^{*})^{\omega}$. We write $u <_{\rm lex,2} v$ if $[u]_2$ is lexicographically smaller than $[v]_2$. We write $u <_{\rm colex,2} v$ if there is a maximal $i$ such that $[u_i]_2 \neq [v_i]_2$, and $[u_i]_2 < [v_i]_2$. Note that while $<_{\rm lex,2}$ is a total order on $(\{0,1\}^{*})^{\omega}$, the order $<_{\rm colex,2}$ is only a partial order. However, $<_{\rm colex,2}$ is a total order on the set of all words $v\in (\{0,1\}^{*})^{\omega}$ such that $[v]_j$ is eventually $0$.
\newline

\noindent Let $u=u_1u_2\cdots,v=v_1v_2\cdots \in (\{0,1\}^{*})^{\omega}$. Let $k$ be  minimal such that $[u_k]_2 \neq [v_k]_2$. We write $u<_{\rm alex,2} v$ if either $k$ is even and $[u_k]_2 < [v_k]_2$, or $k$ is odd and $[u_k]_2 > [v_k]_2$; this is the \textit{alternating lexicographic order} on $(\{0,1\}^{*})^{\omega}$.

\subsection{Ostrowski representations} We now introduce Ostrowski representations based on the continued fraction expansions of real numbers. We refer the reader to Allouche and Shallit \cite{auto_seq} and Rockett and Sz\"{u}sz \cite{RS} for more details.  A \textbf{finite continued fraction expansion} $[a_0;a_1,\dots,a_k]$ is an expression of the form
$$\tiny{
a_0 +\cfrac{1}{a_1 + \cfrac{1}{a_2+ \cfrac{1}{\ddots +  \cfrac{1}{a_k}}}}}
$$
For a real number $\alpha$, we say $[a_0;a_1,\dots,a_k,\dots]$ is a \textbf{continued fraction expansion of $\alpha$} if $\alpha=\lim_{k\to \infty}[a_0;a_1,\dots,a_k]$ and $a_0\in \Z$, $a_i\in \N_{>0}$ for $i>0$. In this situation, we write $\alpha = [a_0;a_1,\dots].$ Every irrational number has precisely one continued fraction expansion, so we will usually refer to \textit{the} continued fraction expansion of a number. We recall the following well-known fact about continued fractions.
	\begin{fact}
	\label{continued_fraction_ordering}
		Let $\alpha=[a_0;a_1,\dots],\alpha'=[a_0';a_1',\dots]\in \R$ be irrational. Let $k\in \N$ be minimal such that $a_k \neq a'_k$. Then $\alpha < \alpha'$ if and only if 
		\begin{itemize}
			\item $k$ is even and $a_k < a'_k$, or
			\item $k$ is odd and $a_k > a'_k$.
		\end{itemize}
	\end{fact}

\noindent For the rest of this subsection, fix a positive irrational real number $\alpha\in (0,1)$ and let $[a_0;a_1,a_2,\dots]$ be the continued fraction expansion of $\alpha$.\newline
\noindent Let $k\geq 1$. A pair $(p_k, q_k)$ is the \textbf{$k$-th convergent of $\alpha$} if $p_k\in \N$, $q_k\in \Z$, $\gcd(p_k,q_k)=1$ and
\[
\frac{p_k}{q_k}  =  [a_0;a_1,\dots,a_k].
\]
Set $p_{-1}:=1,q_{-1}:=0$ and $p_{0}:=a_0,q_{0}:=1$. While formally a pair of integers, in practice we will think of a convergent as the quotient $\frac{p_k}{q_k}$. The convergents satisfy the following equations for $n\geq 1$: 
\begin{align*}
p_{n} = a_{n}p_{n-1} + p_{n-2}, \quad q_{n} &= a_{n}q_{n-1} + q_{n-2}.
\end{align*}

\noindent We now recall a numeration system due to Ostrowski  \cite{Ost}. 

\begin{factC}[{\cite[Ch.~II-\S4]{RS}}]\label{integer_ostrowski}
Let $X \in \N$. Then $X$ can be written uniquely as
\begin{equation}\label{eq:Ost}
X = \sum_{n=0}^{N}b_{n+1}q_{n},
\end{equation}
where $0 \le b_{1}<a_{1}$, $0 \le b_{n+1} \le a_{n+1}$ and $b_{n}=0$ whenever $b_{n+1}=a_{n+1}$.
\end{factC}
	
\noindent 
For $X\in \N$ satisfying \eqref{eq:Ost} we write
\[
X = [b_1b_2\cdots b_Nb_{N+1}]_{\alpha}
\]
and call the word $b_{1}b_{2}\cdots b_{N+1}$ an $\alpha$-Ostrowski representation of $X$. This representation is unique up to trailing zeros. Let $X,Y\in \N$ and let $b_{1}b_{2}\cdots b_{N+1}$ and $c_1c_2\cdots c_{N+1}$ be $\alpha$-Ostrowski representations of $X$ and $Y$ respectively. 
Since Ostrowski representations are obtained by a greedy algorithm, one can see easily that  $X<Y$ if and only if $b_{1}b_{2}\cdots b_{N+1}$  is co-lexicographically smaller than $c_1c_2\cdots c_{N+1}$.\newline

\noindent We now introduce a similar way to represent real numbers, also due to Ostrowski \cite{Ost}. The \textbf{$k$-th difference $\beta_k$ of $\alpha$} is defined as $\beta_k := q_k \alpha - p_k$. We use the following facts about $k$-th differences: for all $n\in \N$
\begin{enumerate}
    \item $\beta_n>0$ if and only if $n$ is even,
    \item $\beta_{0} > -\beta_{1} >  \beta_{2} > -\beta_{3} > \beta_4 >\dots$, and
    \item $-\beta_{n}  =  a_{n+2} \beta_{n+1} + a_{n+4} \beta_{n+3} + a_{n+6} \beta_{n+5} + \dots$ .
\end{enumerate}

\noindent Let $I_{\alpha}$ be the interval $\big[\lfloor \alpha \rfloor - \alpha, 1 + \lfloor \alpha \rfloor - \alpha\big)$. 
	
	\begin{fact}[{cf. \cite[Ch. II.6 Theorem 1]{RS}}]
		\label{real_ostrowski}
		Let $x \in I_{\alpha}$. Then $x$ can be written uniquely as 
		\begin{equation}\label{eq:Ostr}
		\sum_{k=0}^\infty b_{k+1} \beta_k,
		\end{equation}
		where $b_k\in \Z $ with $0\leq b_k \leq a_k$, and $b_{k-1} = 0$ whenever $b_k = a_k$,(in particular, $b_1 \neq a_1$), and $b_k \neq a_k$ for infinitely many odd $k$.
	\end{fact}
	
\noindent For $x\in I_{\alpha}$ satisfying \eqref{eq:Ostr} we write
\[
x = [b_1b_2\cdots]_{\alpha}
\]
and call the infinite word $b_1b_2\cdots$ the $\alpha$-Ostrowski representation of $x$. This is closely connected to the integer Ostrowski representation. Note that for every real number there a unique element of $I_{\alpha}$ such that that their difference is an integer. We define $f_{\alpha} : \R \to I_{\alpha}$ to be the function that maps $x$ to $x-u$, where $u$ is the unique integer such that $x-u\in I_{\alpha}.$ 
	\begin{factC}[{\cite[Lemma 3.4]{H}}]
		\label{z_o_correspondence}
		Let $X \in \mathbb{N}$ be such that $\sum_{k=0}^N b_{k+1} q_k$ is the $\alpha$-Ostrowski representation of $X$. Then 
		\[
		f_{\alpha}(\alpha X) = \sum_{k=0}^\infty b_{k+1} \beta_k
		\]
		is the $\alpha$-Ostrowski representation of $f_{\alpha}(\alpha X)$, where $b_{k+1} = 0$ for $k > N$.
	\end{factC}

\noindent Since $\beta_k > 0 $ if and only if $k$ is even, the order of two elements in $I_{\alpha}$ can be determined by the Ostrowski representation as follows.

	\begin{factC}[{\cite[Fact 2.13]{H}}]
		\label{real_ostrowski_order}
		Let $x,y \in I_{\alpha}$ with $x\neq y$ and  let $[b_1b_2\cdots ]_{\alpha}$ and $[c_1c_2\cdots ]_{\alpha}$ be the $\alpha$-Ostrowski representations of $x$ and $y$. Let $k \in \N$ be minimal such that $b_k \neq c_k$. Then $x<y$ if and only if
		 \begin{itemize}[align=left]
		\item[(i)] $b_{k+1} < c_{k+1}$ if $k$ is even;
		\item[(ii)] $b_{k+1} > c_{k+1}$ if $k$ is odd.
		\end{itemize}
	\end{factC}

\section{\texorpdfstring{$\#$}{Hashtag}-binary encoding}

\noindent In this section, we introduce $\#$-binary coding.  A similar encoding has been used in Hodgson \cite{Hodgson82}. Fix the alphabet $\SigmaH := \{0, 1, \#\}$.  Let $H_{\infty}$ denote the set of all infinite $\SigmaH$-words in which $\#$ appears infinitely many times. Clearly $H_{\infty}$ is $\omega$-regular.\newline

\noindent Let $C_\# : (\{0,1\}^*)^{\omega} \to H_{\infty}$ map an infinite word $b=b_1b_2b_3\cdots$ over $\{0,1\}^*$ to the infinite $\SigmaH$-word
\[
\# b_1 \# b_2 \# b_3 \# \cdots.
\]
We note that the map $C_\#$ is a bijection.\newline

\noindent Let $u=u_1u_2u_3\cdots,v=v_1v_2v_3\cdots \in \SigmaH^{\omega}$. We say $u$ and $v$ are \textbf{aligned} if for all $i\in\N$
\[
u_i=\# \text{ if and only if } v_i =\#.
\]
This defines an $\omega$-regular equivalence relation on $\SigmaH^{\omega}$. We denote this equivalence relation by $\sim_{\#}$. We say $(w_1,\dots,w_n)\in (\SigmaH^{\omega})^n$ is \textbf{aligned} if
\[
w_1 \sim_{\#} w_2 \sim_{\#} \dots \sim_{\#} w_n. 
\]
We say a subset $X\subseteq (\SigmaH^{\omega})^n$ is \textbf{aligned} if every $w\in X$ is aligned.\newline

\noindent The following fact follows easily.

\begin{fact}\label{fact:regularorder}
The following sets are $\omega$-regular:
\begin{itemize}
\item
$\{ (u,v) \in H_{\infty}^2 \ : \ u \sim_{\#} v \text{ and } C_{\#}^{-1}(u) <_{\rm lex,2} C_{\#}^{-1}(v)\}$,
\item
$\{ (u,v) \in H_{\infty}^2 \ : \ u \sim_{\#} v \text{ and } C_{\#}^{-1}(u) <_{\rm colex,2} C_{\#}^{-1}(v)\}$,
\item $\{ (u,v) \in H_{\infty}^2 \ : \ u \sim_{\#} v \text{ and } C_{\#}^{-1}(u) <_{\rm alex,2} C_{\#}^{-1}(v)\}$.
\end{itemize}
\end{fact}

\subsection{\texorpdfstring{$\#$}{Hashtag}-binary coding of continued fractions} 
We now code the continued fraction expansions of real numbers as infinite $\SigmaH$-words.
\begin{defi}
Let $\alpha\in (0, 1)$ be irrational such that $[0;a_1,a_2,\dots]$ is the continued fraction expansion of $\alpha$. Let $u=u_1u_2\cdots\in (\{0,1\}^*)^{\omega}$ such that 
$u_i\in \{0,1\}^*$ is a binary representation of $a_i$ for each $i\in \Z_{\geq 0}$. We say that $C_\#(u)$ is a \textbf{$\#$-binary coding of the continued fraction of $\alpha$}.
\end{defi}

\noindent Let $R$ be the set of elements of $\SigmaH^{\omega}$ of the form $(\#(0|1)^*1(0|1)^*)^\omega$. Obviously, $R$ is $\omega$-regular.

\begin{lem}\label{lem:uniquehash}
Let $w \in R$. Then there is a unique irrational number $\alpha\in [0,1]$ such that $w$ is a 
$\#$-binary coding of the continued fraction of $\alpha$.
\end{lem}
\begin{proof}
By the definition of $R$, there is $w_1w_2\cdots \in ((0|1)^*1(0|1)^*)^\omega$ such that 
\[
w = \#w_1\#w_2\# \cdots.
\]
Since $w_i \in (0|1)^*1(0|1)^*$, we have that $w_i$ is 
a $\{0,1\}$-word containing at least one $1$. Let $a_i$ be the natural number that $a_i =[w_i]_2$. Because $w_i$ contains a $1$, we must have $a_i \neq 0$. Thus $w$ is a $\#$-binary coding of the infinite continued fraction of the irrational $\alpha = [0; a_1, a_2, \dots]$. Uniqueness follows directly from the fact that both binary expansions and continued fraction expansions only represent one number.
\end{proof}

\noindent For $w\in R$, let $\alpha(w)$ be the real number given by Lemma \ref{lem:uniquehash}. When $v=(v_1,\dots,v_n) \in R^n$, we write $\alpha(v)$ for $(\alpha(v_1),\dots, \alpha(v_n))$.\newline

\noindent Even though continued fractions are unique, their $\#$-binary codings are not, because binary representations can have trailing zeroes. This ambiguity is required in order to properly recognize relationships between multiple numbers, as one of the numbers involved may require more bits in a coefficient than the other(s). Occasionally we need to ensure that all possible representations of a given tuple of numbers are contained in a set. For this reason, we introduce the zero-closure of subsets of $R^n$. 

\begin{defi} Let $X\subseteq R^n$ be aligned. The \textbf{zero-closure} of $X$ is 
\[
\{u \in R^n \ : \ u \text{ is aligned } \wedge \  \exists v \in X  \  \alpha(u) = \alpha(v)\}.
\]
\end{defi}
\begin{lem}
\label{zero_closure}
    Let $X \subseteq R^n$ be $\omega$-regular and aligned. Then the zero-closure of $X$ is also $\omega$-regular.
\end{lem}
\begin{proof}
    Let $\mathcal{A}$ be a B\"uchi automaton recognizing $X$. We use $Q$ to denote the set of states of $\cal A$. We create a new automaton $\mathcal{A}'$ that recognizes the zero-closure of $X$, as follows:
    \begin{enumerate}[leftmargin=*,labelindent=16pt,label=\bfseries Step \arabic*.]
    \item[(Step 1)] Start with the automata $\cal A$.
        \item[(Step 2)] For each transition on the $n$-tuple $(\#,\dots,\#)$ from a state $p$ to a state $q$, we add a new state $\mu(p,q)$ that loops to itself on the $n$-tuple $(0,\dots,0)$ and transitions to state $q$ on $(\#,\dots,\#)$. We add a transition from $p$ to $\mu(p,q)$ on $(0,\dots,0)$.
        \item[(Step 3)] For every pair $p, q$ of states of $\cal A$  for which $p$ has a run to $q$ on a word of the form $
        (0,\dots,0)^m (\#,\dots,\#)$ for some $m$,
    we add a transition from state $p$ to a new state $\nu(p,q)$ on $(\#,\dots,\#)$, and for every transition out of state $q$, we create a copy of the transition that starts at state $\nu(p,q)$ instead. If any original run from state $p$ to state $q$ passes through a final state, we make $\nu(p,q)$ a final state.
    \item[(Step 4)] Denote the resulting automaton by $\cal A'$ and its set of states by $Q'$.
    \end{enumerate}
    
\vspace{0.2cm}
    
\noindent    We now show that $L(\cal A')$ is the zero-closure of $X$.  We first show that the zero-closure is contained in $L(\cal A')$. 
    Let $v \in X$ and $w \in R^n$ be such that $w$ is aligned and $\alpha(v) = \alpha(w)$.  Since both $v$ and $w$ are aligned, there are 
 $b=b_1b_2\cdots,c=c_1c_2\cdots  \in ((\{0,1\}^n)^*)^\omega$ such that $C_{\#}(b)=v$ and $C_{\#}(c)=w$. 
 Since $\alpha(v)=\alpha(w)$, we have that $[b_i]_2=[c_i]_2$ for $i \in \N$. Therefore, for each $i\in \N$, the words $b_i$ and $c_i$ only differ by trailing (tuples of) zeroes. Let $s=s_1s_2\cdots \in Q^{\omega}$ be an accepting run of $v$ on $\cal A$. We now transfer this run into an accepting run $s'=s_1's_2'\cdots$ of $w$ on $\mathcal{A}'$. For $i\in \N$, let $y(i)$ be the position of the $i$-th $(\#,\dots,\#)$ in $v$ and let $z(i)$ be the position of the $i$-th $(\#,\dots,\#)$ in $w$. For each $i\in \N$, we define a sequence $s'_{z(i)+1}\cdots s'_{z(i+1)}$ of states of $\cal A'$ as follows: 
 \begin{enumerate}
    \item If $|c_i|=|b_i|$, then $c_i=b_i$. We set
    \[
    s'_{z(i)+1}\cdots s'_{z(i+1)} := s_{y(i)+1}\cdots s_{y(i+1)}.
    \]
     \item If $|c_i|>|b_i|$, then $c_i=b_i(0,\dots,0)^{|c_i|-|b_i|}$. We set
     \begin{align*}
     s'_{z(i)+1}&\cdots s'_{z(i+1)} \\
     &:= s_{y(i)+1}\cdots s_{y(i+1)-1}\underbrace{\mu(s_{y(i+1)-1},s_{y(i+1)}) \cdots \mu(s_{y(i+1)-1},s_{y(i+1)}}_{(|c_i|-|b_i|)\text{-times} }s_{y(i+1)}.
     \end{align*}
    Thus the new run follows the old run up to $s_{y(i+1)-1}$ and then transitions to one of the newly added states in the Step 2. It loops on $(0,\dots,0)$ for $|c_i|-|b_i|-1$-times before moving to $s_{y(i+1)}$.
          \item  If $|c_i|<|b_i|$, then $b_i=c_i(0,\dots,0)^{|b_i|-|c_i|}$. We set
        \[
     s'_{z(i)+1}\cdots s'_{z(i+1)} := s_{y(i)+1}\cdots s_{y(i)+|c_i|}\nu(s_{y(i)+|c_i|},s_{y(i+1)}).  
     \]
     The new run utilizes one of the newly added $(\#,\dots,\#)$ transitions and corresponding states added in Step 3.
 \end{enumerate}
 The reader can now easily check that $s'$ is an accepting run of $w$ on $\cal A'$.\newline
 
 \noindent We now show that $L(\cal A')$ is contained in the zero-closure of $X$. We prove that the only accepting runs on $\mathcal{A}'$ are based on accepting runs on $\mathcal{A}$ with trailing zeroes either added or removed. Let $w=w_1w_2\cdots \in L(\cal A')$, 
 and let $s'=s'_1s'_2\cdots \in Q'^{\omega}$ be an accepting run of $w$ on $\cal A'$. We construct $v\in X$ and a run $s=s_1s_2\cdots\in Q^{\omega}$ of $w_2$ on $\mathcal{A}$ such that $\alpha(v) = \alpha(w)$ and $s$ is an accepting run of $v$. We start by setting
 $
 v := w_1w_2\cdots \hbox{ and } s := s_1's_2'\cdots .
 $
For each $i\in \N$, we replace $w_i$ in $v$ and $s_i'$ in $s$ as follows:
 \begin{enumerate}
     \item If $s_i'\in Q$, then we make no changes to $s_i'$ and $w_i$.
    \item If $s_i'=\mu(p,q)$ for some $p,q\in Q$, we delete the $s_i'$ in $s$ and delete $w_{i}$ in $v$.
     \item If $s_i=\nu(p,q)$ for some $p,q\in Q$, then we replace 
     \begin{itemize}
         \item[(a)] $s_i'$ by a run $t=t_1\cdots t_{n+1}$ of $(0,\dots,0)^n(\#,...,\#)$ from $p$ to $q$, and 
         \item[(b)] $w_i$ by $(0,\dots,0)^n(\#,...,\#)$.
     \end{itemize}  If $\nu(p,q)$ is a final state of $\cal A'$, we choose $t$ such that it passed through a final state of $\cal A$.
  \end{enumerate}
It is clear that the resulting $s$ is in $Q^{\omega}$. The reader can check $s$ is an accepting run of $v$ on $\cal A$ and that $\alpha(v)=\alpha(w)$. Thus $w$ is in the zero-closure of $X$.
   \end{proof}

	\begin{lem}
	\label{continued_fraction_ordering_reg}
	The set 
\[
\{ (w_1,w_2)\in R^2 \ : \ w_1 \sim_\# w_2 \text{ and } \alpha(w_1)<\alpha(w_2) \} 
\]
is $\omega$-regular.
	\end{lem}
\begin{proof}
Let $w_1,w_2\in R$ be such that $w_1 \sim_\# w_2$. 	
By Fact \ref{continued_fraction_ordering} we have that  $\alpha(w_1)<\alpha(w_2)$ if only $C_\#^{-1}(w_1) <_{\rm alex,2} C_\#^{-1}(w_2)$. Thus $\omega$-regularity follows from Fact \ref{fact:regularorder}.
\end{proof}

	\begin{lem}
		\label{continued_fraction_quadratic_reg}
		Let $a \in [0, 1)$ be a quadratic irrational. Then 
		\[
		\{ w \in R \ : \ \alpha(w)=a\}
		\]
		 is $\omega$-regular.
	\end{lem}
	\begin{proof}
		The continued fraction expansion of $a$ is eventually periodic (see for example \cite[Theorem 177]{HW}).
        Thus there is an eventually periodic $u\in (\{0,1\}^*)^{\omega}$ such that $C_{\#}(u)$ is a $\#$-binary coding of the continued fraction of $a$. The singleton set containing an eventually periodic string is $\omega$-regular. It remains to expand this set to contain all representations via Lemma~\ref{zero_closure}.
	\end{proof}

\begin{lem}
\label{less_than_half}
The set $\{ w \in R \ : \ \alpha(w) < \frac12\}$ is $\omega$-regular.
\end{lem}	
\begin{proof}
Let $\alpha(w) = [0; a_1, a_2, \dots]$.
It is easy to see that $\alpha(w) < \frac{1}{2}$ if and only if $a_1 > 1$.
Thus we need only check that $a_1 \neq 1$. The set of $w\in R$ for which this true is just $R \setminus Y$, where $Y\subseteq \SigmaH^{\omega}$ is given by the regular expression $\#10^*(\#(0 \cup 1)^*)^\omega$.
\end{proof}

\subsection{\texorpdfstring{$\#$}{Hashtag}-Ostrowski-representations} We now extend the $\#$-binary coding to Ostrowski representations.

\begin{defi}
Let $v,w\in (\SigmaH)^{\omega}$, let $x=x_1x_2x_3\cdots \in \N^{\omega}$ and let $b=b_1b_2b_3 \cdots \in (\{0,1\}^*)^{\omega}$ be such that $w=C_\#(b)$ and $[b_i]_2 = x_i$ for each $i$. 
\begin{itemize}
    \item For $N\in \N$, we say that $w$ is a \textbf{$\#$-$v$-Ostrowski representation of $N$} if $v$ and $w$ are aligned and $x$ is an $\alpha(v)$-Ostrowski representation of $N$.
    \item For $c\in I_{\alpha(v)}$, we say that $w$ is a \textbf{$\#$-$v$-Ostrowski representation of $c$} if $v$ and $w$ are aligned and $x$ is an $\alpha(v)$-Ostrowski representation of $c$.
\end{itemize}
We let $A_{v}$ denote the set of all words $w \in \SigmaH^{\omega}$ such that 
$w$ is a $\#$-$v$-Ostrowski representation of some $c\in I_{\alpha(v)}$, and similarly, 
by $A_{v}^{\rm fin}$ the set of all words $w \in \SigmaH^{\omega}$ such that 
$w$ is a $\#$-$v$-Ostrowski representation of some $N\in \N$.
\end{defi}

\begin{lem}\label{afin_a_regular}
The sets 
\[
A^{\rm fin} :=\{ (v,w) \ : \ v \in R, w \in A_{v}^{\rm fin}\}, \hbox{ and } A :=\{ (v,w) \ : \ v \in R, w \in A_{v}\}.\]
are $\omega$-regular. Moreover, $A^{\rm fin} \subseteq A$.
\end{lem}
\begin{proof}
The statement that $A^{\rm fin} \subseteq A$, follows immediately from the definitions of $A^{\rm fin}$ and $A$ and Fact \ref{z_o_correspondence}. It is left to establish the $\omega$-regularity of the two sets.\newline

\noindent For $A^{\rm fin}$: Let $B\supseteq A^{\rm fin}$ be the set of all pairs $(v,w)$ such that $v\in R$ and $v\sim_{\#} w.$
Note that $B$ is $\omega$-regular. Let $(v,w)\in B$.
Since  $v$ and $w$ have infinitely many $\#$ symbols and are aligned,
there are unique $a=a_1a_2\cdots, b=b_1b_2\cdots \in (\{0,1\}^*)^{\omega}$ such that $C_\#(a)=v$, $C_\#(b)=w$ and $|a_i|=|b_i|$ for each $i\in \N$.
Then by Fact \ref{integer_ostrowski}, $(v, w) \in A^{\rm fin}$ if and only if
\begin{enumerate}[label=(\alph*)]
        \item $b$ has finitely many $1$ symbols;
        \item $b_1 <_{\rm colex} a_1$;
        \item $b_i \leq_{\rm colex} a_i$ for all $i > 1$;
        \item if $b_i = a_i$, then $b_{i-1} = 0$.
    \end{enumerate}
It is easy to check that all four conditions are $\omega$-regular. \newline

\noindent For $A$: As above, let $(v,w)\in B$.
Since  $v$ and $w$ have infinitely many $\#$ symbols and are aligned,
there are unique $a=a_1a_2\cdots, b=b_1b_2\cdots \in (\{0,1\}^*)^{\omega}$ such that $C_\#(a)=v$, $C_\#(b)=w$ and $|a_i|=|b_i|$ for each $i\in \N$.
 Then by Fact \ref{real_ostrowski}, $(v, w) \in A$ if and only if
    
    \begin{enumerate}[label=(\alph*),start=5]
        \item $b_1 <_{\rm colex} a_1$;
        \item $b_i \leq_{\rm colex} a_i$ for all $i > 1$;
        \item if $b_i = a_i$, then $b_{i-1} = 0$;
        \item $b_i \neq a_i$ for infinitely many odd $i$.
    \end{enumerate}
Again, it is easy to see that all four conditions are $\omega$-regular.
\end{proof}

\begin{defi}
Let $v\in R$. We define $Z_v : A_v^{\rm fin} \to \N$ to be the function that maps $w$ to the natural number whose $\#$-$v$-Ostrowski representation is $w$.\newline
Similarly, we define $O_v:A_v \to I_{\alpha(v)}$ to be the function that maps $w$ to the real number  whose $\#$-$v$-Ostrowski representation is $w$.
\end{defi}

\begin{lem}\label{zv_ov_bijective} Let $v \in R$. Then $Z_v : A_v^{\rm fin} \to \mathbb{N}$ and $O_v : A_v \to I_{\alpha(v)}$ are bijective.
\end{lem}

\begin{proof}
We first consider injectivity. By Fact \ref{integer_ostrowski} and Fact \ref{real_ostrowski} a number in $\mathbb{N}$ or in $I_{\alpha(v)}$ only has one $\alpha(v)$-Ostrowski representation. So we only need to explain why such a representation will only have one encoding in $A_v^{\rm fin}$ (respectively $A_v$). This follows from the uniqueness of binary representations up to the length of the representation, and from the fact that the requirement of having the $\#$ symbols aligned with $v$ determines the length of each binary-encoded coefficient.\newline

\noindent For surjectivity we only need to explain why an $\alpha(v)$-Ostrowski representation can always be encoded into a string in $A_v^{\rm fin}$ (respectively $A_v$). It suffices to show that the requirement of having the $\#$ symbols aligned with $v$ will never result in needing to fit the binary encoding of a number into too few symbols, i.e., that it will never result in having to encode a natural number $n$ in binary in fewer than $1 + \lfloor \log_2 n \rfloor$ symbols. Since the function $1 + \lfloor \log_2 n \rfloor$ is monotone increasing, we can encode any natural number below $n$ in $k$ symbols if we can encode $n$ in binary in $k$ symbols. However, by Fact \ref{integer_ostrowski} and Fact \ref{real_ostrowski}, the coefficients in an $\alpha(v)$-Ostrowski representation never exceed the corresponding coefficients in the continued fraction for $\alpha(v)$, i.e., $b_n \leq a_n$.
\qedhere

\end{proof}
\begin{defi} Let $v\in R$. We write $\mathbf{0}_v$ for $Z_v^{-1}(0)$, and $\mathbf{1}_v$ for $Z_v^{-1}(1)$. 
\end{defi}

\begin{lem}\label{0_1_regular}
The relations $\mathbf{0}_* = \{(v,\mathbf{0}_v) \ : \ v \in R\}$ and $\mathbf{1}_* = \{(v,\mathbf{1}_v) \ : \ v \in R\}$ are $\omega$-regular.
\end{lem}
\begin{proof}
    Recognizing $\mathbf{0}_*$ is trivial, as the Ostrowski representations of $0$ are of the form $0\cdots 0$ for all irrational $\alpha$. Thus $\mathbf{0}_*$ is just the relation $$\{(v,w) \ : v \in R, w \text{ is $v$ with all $1$ bits replaced by $0$ bits}\}.$$ This is clearly $\omega$-regular.\newline 

    \noindent We now consider $\mathbf{1}_*$. Let $\alpha =[0;a_1,a_2,\dots]$ be an irrational number.
    If $a_1 > 1$, the $\alpha$-Ostrowski representations of $1$ are of the form $10\cdots0$. If $a_1 = 1$, the $\alpha$-Ostrowski representations of $1$ are of the form $010\cdots0$. Thus, in order to recognize $\mathbf{1}_*$, we only need to be able to recognize if a number in binary representation is $0$, $1$, or greater than $1$. Of course, this is easily done on a B\"uchi automaton. 
\end{proof}

\begin{lem}\label{lem:Oofone} Let $s\in A_v^{\rm fin}$. Then $\alpha(v) Z_{v}(s) - O_v(s)\in \Z$ and \[
O_v(\mathbf{1}_v) = \begin{cases} \alpha(v) & \hbox{ if } \alpha(v)< \frac{1}{2};\\
\alpha(v) - 1 & \hbox{otherwise.}
\end{cases}
\]
\end{lem}

\begin{proof}
By Fact \ref{z_o_correspondence}, $O_v(s) = f_{\alpha(v)}(\alpha(v) Z_v(s))$. Thus 
\[
\alpha(v) Z_{v}(s) - O_v(s) = \alpha(v) Z_{v}(s) - f_{\alpha(v)}(\alpha(v) Z_{v}(s)),
\]which is an integer by the definition of $f$. By the definition of $\mathbf{1}_v$ and by Fact \ref{z_o_correspondence}, we know $O_v(\mathbf{1}_v) = f_{\alpha(v)}(\alpha(v))$ is the unique element of $I_{\alpha(v)}$ that differs from $\alpha(v)$ by an integer. If $0 < \alpha(v) < \frac{1}{2}$, then 
\[
-\alpha(v) < \alpha(v) < 1 - \alpha(v).
\]
Thus in this case, $\alpha(v) \in I_{\alpha(v)}$ and $O_v(\mathbf{1}_v) = \alpha(v)$.
When $\frac{1}{2} < \alpha(v) < 1$, then 
\[
-\alpha(v) < \alpha(v) - 1 < 1 - \alpha(v).
\]
Therefore $\alpha(v) - 1 \in I_{\alpha(v)}$ and $O_v(\mathbf{1}_v) = \alpha(v)-1$.
\qedhere
\end{proof}

\begin{lem}\label{lem:order-reg}
The sets 
\begin{align*}
 \prec^{\rm fin} &:= \{(v,s,t) \in \SigmaH^3 \ : \ s,t\in A_v^{\rm fin} \wedge  Z_v(s) < Z_v(t)\},\\
 \prec &:= \{(v,s,t) \in \SigmaH^3 \ : \ s,t \in A_v \wedge O_v(s) < O_v(t)\}
\end{align*}
are $\omega$-regular.
\end{lem}
\begin{proof}
For $\prec^{\rm fin}$, first recall that for $X,Y\in \N$ and $\alpha$ irrational, we have $X<Y$ if and only if the $\alpha$-Ostrowski representation of $X$ is co-lexicographically smaller than the $\alpha$-Ostrowski representation of $Y$. Therefore, we need only recognize co-lexicographic ordering on the list of coefficients, with each coefficient ordered according to binary. This follows immediately from Fact \ref{fact:regularorder}.\newline

\noindent For $\prec$, note that by Fact \ref{real_ostrowski_order} the usual order on real numbers corresponds to the alternating lexicographic ordering on real Ostrowski representations. Therefore, we need only recognize the alternating lexicographic ordering on the list of coefficients, with each coefficient ordered according to binary. This follows immediately from Fact \ref{fact:regularorder}.
\end{proof}

\noindent We consider $\R^n$ as a topological space using the usual order topology. For $X\subseteq \R^n$, we denote its topological closure by $\overline{X}$. This is of course defined using the product of order topologies; i.e. $x \in \overline{X}$ iff every open box containing $x$ also contains an element of $X$.

\begin{cor}\label{cor:closure}
Let $W\subseteq (\SigmaH^{n+1})^*$ $\omega$-regular be such that $$W \subseteq \{ (v,s_1,\dots,s_n) \in (\SigmaH^{n+1})^* \ : \ s_1,\dots, s_n \in A_v\}.$$
Then the following set is also $\omega$-regular:
\[
\overline{W} := \{ (v,s_1,\dots,s_n) \in (\SigmaH^{n+1})^* \ : \ s_1,\dots, s_n \in A_v \wedge (O_v(s_1),\dots, O_v(s_n)) \in \overline{O(W_v)}\}.
\]
\end{cor}
\begin{proof}
    Let $(v,s_1,\dots,s_n) \in (\SigmaH^{n+1})^*$ be such that $s_1,\dots, s_n \in A_v$. Let $X_i = O_v(s_i)$. By the definition of the topological closure, we have that  $(X_1, \dots, X_n) \in \overline{O(W_v)}$ if and only if
    for all $Y_1,\dots Y_n,Z_1, \dots , Z_n\in \R$ with $Y_i<X_i <Z_i$ for $i=1,\dots,n$ there are $X'=(X_1',\dots, X_n')\in O(W_v)$ such that
$Y_i<X_i' <Z_i$ for $i=1,\dots,n$.  Thus by Lemma \ref{lem:order-reg}, $(v,s_1,\dots,s_n) \in \overline{W}$ if and only if for all $t_1,\dots t_n,u_1, \dots , u_n\in A_v$ with $t_i\prec s_i \prec u_i$, there are $s'=(s_1',\dots, s_n')\in W_v$ such that
$t_i\prec s_i' \prec u_i$ for $i=1,\dots,n$. The latter condition is $\omega$-regular by Fact \ref{fact:omegaclosure}.
\end{proof}

\section{Recognizing addition in Ostrowski numeration systems}
\label{section:addition}

\noindent The key to the rest of this paper is a general automaton for recognizing addition of Ostrowski representations uniformly. We will prove the following:

\begin{thm}
		\label{oplus_fin_regular}
		The set 
\[
\oplus^{\rm fin} := \{ (v, s_1, s_2,s_3) \ : \ s_1,s_2,s_3 \in A^{\rm fin}_v \wedge Z_v(s_1) + Z_v(s_2)=Z_v(s_3) \}
\]
is $\omega$-regular.
	\end{thm}

\noindent In order to prove this theorem, we will introduce a method to generate more complex automata for strings in $H_\infty$, from general B\"uchi automata. For the reasons mentioned when general B\"uchi automata were introduced in Section \ref{section:prelim}, we will not use these automata directly. Instead, we will use the $\#$-binary coding to convert the computation to a more familiar setting. Similarly arguments have been made before, in particular in \cite[Section 4]{Hodgson82}.

\begin{defi}
    Let $w=w_1w_2\cdots \in (\N^n)^\omega$. A \textbf{$\#$-binary coding of $w$} is a word $u = u_1 u_2 \dots \in (\SigmaH^n)^{\omega}$ such that
    \[
    C_{\#}(u_{1,i}u_{2,i}\cdots) = w_{1,i}w_{2,i}\cdots,
    \]
    where $u_{j,i}$ and $w_{j,i}$ denote the $i$-th component of $j$-th character of $u$ and $w$.\newline
    Let $X \subseteq (\N^n)^\omega$. The language of \textbf{$\#$-binary coding of $X$} is the set of all $\#$-binary codings of its elements.
\end{defi}


\begin{lem}\label{lem:binary-coding-infinite-aut}
    Let $\cal A=(Q,\N^n,\Delta,I,F)$ be a general B\"uchi automaton over $\N^n$, possibly with infinitely many transitions, such that for every $s_1,s_2 \in Q$ the 
        set 
\[
\{ u \in \{0,1\}^* \ : \ (s_1,[u]_2,s_2) \in \Delta\} 
\]        
is regular. Then the $\#$-binary coding of the language accepted by $\cal A$ is $\omega$-regular.
\end{lem}
\begin{proof}
    \noindent We construct from $\cal A$ a new B\"uchi automaton $\cal A'$ over $(\Sigma_\#)^n$. It is constructed via the following procedure:

    \begin{enumerate}
        \item Copy the states (without their transitions) from $\cal A$ to $\cal A'$. Any final states in $\cal A$ are to remain final in $\cal A'$.
        \item Add an initial state $q_{start}$, and endow it with transitions to every state that was an initial state in $\cal A$ on the character $(\#,\dots,\#)$. These states are no longer initial in $\cal A'$, so that $q_{start}$ is the only initial state.
        \item For every pair $s_1, s_2\in Q$:
        \begin{enumerate}
            \item Let $\cal B$ be a finite automaton recognizing
            \[
\{ u \in \{0,1\}^* \ : \ (s_1,[u]_2,s_2) \in \Delta\}.
\] Add the states and transitions of $\cal B$ to $\cal A'$.
            \item For every initial state $t$ in $\cal B$, whenever $t$ transitions to $t'$ on a character, add a transition from $s_1$ to $t'$ on the same character. Make $t$ no longer an initial state in $\cal A'$.
            \item For every final state $t$ in $\cal B$, add a transition from $t$ to $s_2$ on $(\#,\dots,\#)$. Make $t$ no longer a final state in $\cal A'$.
            \item If the empty word $\epsilon$ was accepted by $\cal B$, then add a transition from $q$ to $r$ on $(\#,\dots,\#)$.
        \end{enumerate}
    \end{enumerate}

    \noindent One can check that the language accepted by $\cal A'$ is the $\#$-binary coding of the language accepted by $\cal A$. Indeed, if a word is accepted by $\cal A'$, it must begin with $\#^n$ and be followed by a sequence of binary codings that correspond to transitions in $\cal A$, delimited by $\#$, and visiting final states of $\cal A$ infinitely often.
\end{proof}

\noindent We will illustrate with an example. Figure \ref{fig:automaton-insertion} demonstrates the process of applying Lemma \ref{lem:binary-coding-infinite-aut} to a simple automaton that accepts any infinite string of natural numbers containing at least one odd number. \newline

\begin{figure}[p]
    \centering
    \begin{tikzpicture}[shorten >=1pt,node distance=5cm,on grid,auto] 
        \node[state,initial,minimum size=5em] (q0)   {$q_0$}; 
        \node[state,accepting,minimum size=5em] (q1) [right=of q0] {$q_1$};
        \path[->] 
            (q0) 
                edge [loop above] node {even} ()
                edge node {odd} (q1)
            (q1)
                edge [loop above] node {any} ()
            ;
    \end{tikzpicture}

    (a)

    \vspace{1cm}
    
    \begin{tikzpicture}[shorten >=1pt,node distance=2cm,on grid,auto] 
        \node[state,initial] (qe0)   {$q_{even,0}$}; 
        \node[state,accepting] (qe1) [right=of qe0] {$q_{even,1}$};
        \node[state,initial] (qo0) [right=3cm of qe1] {$q_{odd,0}$}; 
        \node[state,accepting] (qo1) [right=of qo0] {$q_{odd,1}$};
        \node[state,initial,accepting] (qa0) [right=3cm of qo1] {$q_{any,0}$}; 
        \path[->] 
            (qe0) 
                edge node {0} (qe1)
            (qe1)
                edge [loop above] node {0,1} ()
            (qo0) 
                edge node {1} (qo1)
            (qo1)
                edge [loop above] node {0,1} ()
            (qa0)
                edge [loop above] node {0,1} ()
            ;
    \end{tikzpicture}
    
    (b)

    \vspace{1cm}
    
    \begin{tikzpicture}[shorten >=1pt,node distance=2cm,on grid,auto] 
        \node[state,initial] (qs)   {$q_{start}$}; 
        \node[state,minimum size=5em] (q0) [right=3cm of qs] {$q_0$}; 
        \node[state] (qe0) [above=3cm of q0] {$q_{even,0}$}; 
        \node[state] (qe1) [right=of qe0] {$q_{even,1}$};
        \node[state] (qo0) [right=3cm of q0] {$q_{odd,0}$}; 
        \node[state] (qo1) [right=of qo0] {$q_{odd,1}$};
        \node[state,accepting,minimum size=5em] (q1) [right=3cm of qo1] {$q_1$};
        \node[state] (qa0) [above=3cm of q1] {$q_{any,0}$}; 
        \path[->] 
            (qs)
                edge node {$\#$} (q0)
            (q0)
                edge [bend left=10] node {0} (qe1)
                edge [bend right=45] node {1} (qo1)
            (qe0) 
                edge node {0} (qe1)
            (qe1)
                edge [loop above] node {0,1} ()
                edge [bend left=10] node {$\#$} (q0)
            (qo0) 
                edge node {1} (qo1)
            (qo1)
                edge [loop above] node {0,1} ()
                edge node {$\#$} (q1)
            (q1)
                edge [bend left=10] node {0,1} (qa0)
                edge [loop right] node {$\#$} ()
            (qa0)
                edge [loop above] node {0,1} ()
                edge [bend left=10] node {$\#$} (q1)
            ;
    \end{tikzpicture}
    
    (c)

    \vspace{1.5cm}

    \caption{The procedure of Lemma \ref{lem:binary-coding-infinite-aut}. (a) The original automaton, with transitions for ``any even number,'' ``any odd number,'' and ``any number.'' (b) The finite automata recognizing these sets in binary encoding. (c) The combined automaton produced by Lemma \ref{lem:binary-coding-infinite-aut}.}

    \label{fig:automaton-insertion}
\end{figure}

\noindent We may now give the full proof of Theorem \ref{oplus_fin_regular}.

\begin{proof}[Proof of Theorem \ref{oplus_fin_regular}]
\noindent In \cite[Section 2]{BSS} 
the authors generate a general finite automaton $\mathcal{A}_0$ over the alphabet $\N^4$ such that a finite word  $(d_1,x_1,y_1,z_1)(d_2,x_2,y_2,z_2)\cdots$ $(d_m,x_m,y_m,z_m)\in (\N^4)^*$ is accepted by $\mathcal{A}_0$ if and only if there are $d_{m+1},\ldots  \in \N$ and $x,y,z\in \N$ such that for $\alpha = [0;d_1,d_2,\dots]$
we have 
\begin{align*}
    x&=[x_1x_2\cdots x_m]_\alpha \\
    y&=[y_1y_2\cdots y_m]_\alpha \\
    z&=[z_1z_2\cdots z_m]_\alpha\\
    z&=x+y.
\end{align*}
Let $\mathcal{A}$ be the general Büchi automaton with the same underlying quintuple as $\mathcal A_0$. It follows immediately that
if $(d_1,x_1,y_1,z_1)(d_2,x_2,y_2,z_2)\cdots\in (\N^4)^{\omega}$ is accepted by $\mathcal{A}$ if and only if there is an infinite subset $U\subseteq \N$ such that for all $u\in \N$
\[
[x_1x_2\cdots x_u]_\alpha + [y_1y_2\cdots y_u]_\alpha =[z_1z_2\cdots z_u]_\alpha
\]
Each transition in $\mathcal{A}$ corresponds to a linear equation with constant integer coefficients. As an example, one of the transitions in Figure 3 of \cite{BSS} is given as ``$-d_i + 1$,'' meaning that it represents all cases where, letting $v_{i}, s_{1i}, s_{2i}, s_{3i}$ be the $i$th letter of $v, s_1, s_2, s_3$ respectively, we have $s_{3i} - s_{1i} - s_{2i} = -v_i + 1$. Note that the binary representation of the graph of addition and subtraction, as well as of the constant $1$, are regular. Thus $\mathcal{A}$ satisfies the conditions of Lemma \ref{lem:binary-coding-infinite-aut}. Let $X\subseteq (\SigmaH^4)^{\omega}$ be the $\#$-binary coding of the language accepted by $\mathcal{A}$. By Lemma \ref{lem:binary-coding-infinite-aut}, we know that $X$ is $\omega$-regular. Observe that
\[
\oplus^{\rm fin} = \{ (v, s_1, s_2,s_3)\in X \ : \ s_1,s_2,s_3 \in A^{\rm fin}_v \}
\]
and hence $\omega$-regular.
\qedhere
\end{proof}

\noindent The automaton constructed above has $82$ states\footnote{Schmitthenner \cite{Schmitthenner-thesis} constructs an B\"uchi automaton with just 24 states accepting the same language.}. Using our software Pecan, we can formally check that this automaton recognizes the set in Theorem \ref{oplus_fin_regular}. Following a strategy already used in Mousavi, Schaeffer, and Shallit \cite[Remark 2.1]{MSS1} we check that our adder satisfies the standard inductive definition of addition on the natural numbers; that is, for all $x,y \in \N$
\begin{align*}
    0 + y &= y \\
    s(x) + y &= s(x + y)
\end{align*}
where $x,y\in\N$ and $s(x)$ denotes the successor of $x$ in $\N$. The successor function on $\N$ can be defined using only $<$ as follows:
\[
s(x)=y \hbox{ if and only if } (x<y) \wedge (\forall z \ (z\leq x) \vee (z\geq y)).
\]
Thus in Pecan we define  \pecaninline{bco\_succ(a,x,y)} as
\begin{pecan}
bco_succ(a,x,y) := bco_valid(a,x) & bco_valid(a,y) 
    & bco_leq(x,y) & !bco_eq(x,y)
    & forallz. if bco_valid(a,z) then (bco_leq(z,x) | bco_leq(y,z))
\end{pecan}
where
\begin{itemize}
    \item \pecaninline{bco_eq} recognizes $\{ (x,y) : x = y \}$,
    \item   \pecaninline{bco_leq} recognizes $\{ (x,y) : x \leq_{\rm colex} y\}$, and 
    \item  \pecaninline{bco_valid} recognizes $A_{\rm fin}$.
\end{itemize}
We now confirm that our adder satisfies the above equations using the following Pecan code:
\begin{pecan}
Let x,y,z be ostrowski(a).
Theorem ("Addition base case (0 + y = y).", {
    foralla. forallx,y,z. if bco_zero(x) 
               then (bco_adder(a,x,y,z) iff bco_eq(y,z)) }).
Theorem ("Addition inductive case (s(x) + y = s(x + y)).", {
    foralla. forallx,y,z,u,v. if (bco_succ(a,u,x) & bco_succ(a,v,z)) 
                then (bco_adder(a,x,y,z) iff bco_adder(a,u,y,v)) }).
\end{pecan}
\break 
In the above code  
\begin{itemize}
    \item \pecaninline{bco_adder} recognizes $\oplus^{\rm fin}$,
\item  \pecaninline{bco_zero} recognizes $\mathbf{0}_{*}$, and
    \item \pecaninline{bco_succ} recognizes $\{ (v, x, y) : x,y \in A_v^{\rm fin}, Z_v(x) + 1 = Z_v(y) \}$.
\end{itemize}
Pecan confirms both statements are true. This proves Theorem \ref{oplus_fin_regular} modulo correctness of Pecan and the correctness of the implementations of the automata for \pecaninline{bco_eq},
\pecaninline{bco_leq}, \pecaninline{bco_valid} and \pecaninline{bco_zero}. For more details about Pecan, see Section \ref{section:pecan}.\newline

    \noindent We need the following well-known consequence of K\"onig's Lemma (compare the proof of \cite[Lemma 4.3]{BS2022}).
\begin{fact}
\label{limitlanguage}
    Let $\mathcal{A}$ be a B\"uchi automaton over $\Sigma$ with all states accepting, let $w\in \Sigma^{\omega}$, and let $(u_n)_{n\in \N}$ be a sequence of words in $\Sigma^{\omega}$ such that $u_n|_n=w|_n$ for all $n\in \N$. If $u_n\in L(\mathcal{A})$ for every $n\in \N$, then $w \in L(\mathcal{A})$.
\end{fact}

	\noindent Using this result, we can extend the automaton in Theorem \ref{oplus_fin_regular} to an automaton for addition modulo $1$ on $I_{\alpha}$.
	
	\begin{lem}
		\label{oplus_regular}
	 The set
\[
\oplus:= \{ (v,s_1,s_2,s_3) : s_1,s_2,s_3\in A_v \wedge O_v(s_1) +O_v(s_2)\equiv  \modd{O_v(s_3)}{1}\}
\]
	 is $\omega$-regular. Moreover, $\oplus^{\rm fin} \subseteq \oplus$.
	\end{lem}

\begin{proof}
		First, let $v,s_1, s_2,s_3$ be such that $s_1,s_2,s_3 \in A^{\rm fin}_v$. We claim that on this domain, $(s_1,s_2,s_3) \in \oplus_v$ if and only if $(s_1,s_2,s_3) \in \oplus_v^{\rm fin}$. By Fact~\ref{z_o_correspondence} we know that for all $s\in A_v^{\rm fin}$
		\begin{equation}\label{eq:radd}
		\alpha(v) Z_v(s) - O_v(s)\equiv \modd{0}{1}.
		\end{equation}
		Let $(s_1,s_2,s_3) \in \oplus_v^{\rm fin}$. Then by \eqref{eq:radd}
		\begin{align*}
		    O_v(s_3) &\equiv \modd{\alpha(v) Z_v(s_3)}{1}\\
		    &= \alpha(v) Z_v(s_1) + \alpha(v) Z_v(s_2)\\
		    &\equiv \modd{O_v(s_1) + O_v(s_2)}{1}.
		\end{align*}
Thus $(s_1,s_2,s_3)\in \oplus_v$.\newline

	
\noindent Let $\mathcal{B}^{\rm fin}$ be a B\"uchi automaton recognizing $\oplus^{\rm fin}$. Assume that $\mathcal{B}^{\rm fin}$ is trim. Let $\mathcal{B}'$ be the automaton $\mathcal{B}^{\rm fin}$, but with all states made accepting. Let $S$ be the language accepted by $\mathcal{B}'$. We will show that $S_v \cap A_v^3 = \oplus_v$. Towards that goal, let $(v, s_1, s_2, s_3)\in (\SigmaH^{\omega})^4$ be such that $(s_1, s_2, s_3) \in A_v^3$. It is left to prove that $(s_1, s_2, s_3) \in \oplus_v$ if and if $(s_1,s_2,s_3)\in S_v$.\newline

\noindent Suppose first that $(s_1, s_2, s_3) \in \oplus_v$. Then
\[
O_v(s_{3}) \equiv O_v(s_{1}) + O_v(s_{2})\pmod{1}.
\]
The reader can check using properties of Ostrowski representations that there is a sequence $(s_{m,1},s_{m,2},s_{m,3})_{m\in \N}$ of elements of $ (A_v^{\rm fin})^3$ such that
\begin{enumerate}
    \item $O_v(s_{m,3}) \equiv O_v(s_{m,1}) + O_v(s_{m,2})\pmod{1}$.
    \item $s_{m,i}|_m = s_i|_m$ for $i\in\{1,2,3\}$; i.e., the first $m$ letters of $s_{m,i}$ agree with the first $m$ letters of $s_i$ for $i\in\{1,2,3\}$.
\end{enumerate}
By $(A_v^{\rm fin})^3\cap \oplus_v = \oplus_v^{\rm fin}$ and (1), we know that $(v,s_{m,1},s_{m,2},s_{m,3})$ is accepted by $\mathcal{B}^{\rm fin}$. By Fact \ref{limitlanguage} and (2), we deduce that $\mathcal{B}'$ accepts $(v, s_1, s_2, s_3)$. Thus $(s_1,s_2,s_3)\in S_v$.\newline

\noindent Suppose now that $(s_1, s_2, s_3) \in S_v$. Then $(v,s_1,s_2,s_3)$ is accepted by $\mathcal{B}'$.
For $m\in \N$ and $i\in\{1,2,3\}$,  let $w_{m,i}\in \SigmaH^*$ be such that $w_{m,i}$ is $s_i$ up through the $(m+1)$-st occurrence of $\#$. Thus $w_{m,i}$ represents the first $m$-th Ostrowski coefficients of $O_v(s_i)$. Since $\mathcal{B}^{\rm fin}$ is trim, there exist infinite extensions $s_{m,1}, s_{m,2}, s_{m,3}\in \SigmaH^{\omega}$ of $w_{m,1}, w_{m,2}, w_{m,3}$ such that $\mathcal{B}^{\rm fin}$ accepts $(v, s_{m,1}, s_{m,2}, s_{m,3})$. We now set 
\[
(x_m,y_m,z_m) := (O_v(s_{m,1}),O_v(s_{m,2}),O_v(s_{m,3})), \ (x, y, z) := O_v(s_1, s_2, s_3).
\]
It follows from Fact \ref{real_ostrowski_order} that 
\[
\lim_{m\to \infty}(x_m,y_m,z_m) = (x,y,z).
\]
Because $x_m + y_m \equiv z_m \pmod{1}$ for every $m\in \N$ (by definition of $\mathcal{B}^{\rm fin}$ and $(A_v^{\rm fin})^3\cap \oplus_v = \oplus_v^{\rm fin}$), we have $x + y \equiv z \pmod{1}$. Hence $(s_1, s_2, s_3) \in \oplus_v$.
\end{proof}

\section{The uniform \texorpdfstring{$\omega$-}{omega-}regularity of \texorpdfstring{$\cal R_{\alpha}$}{R-alpha}}
In this section, we turn to the question of the decidability of the logical first-order theory of $\cal R_{\alpha}$. Recall that $\cal R_{\alpha} :=(\mathbb{R},<,+,\Z,\alpha \Z)$ for $\alpha \in \R$. The main result of this section is the following:

\begin{thm}
		\label{b_ralpha_2half}
		There is a uniform family of $\omega$-regular structures $(\mathcal{D}_v )_{v\in R}$
		 such that
		$
		\mathcal{D}_v \simeq \mathcal{R}_{\alpha(v)}
		$ for each $v\in R$.
\end{thm}

\noindent Theorem \ref{b_ralpha_2half} then hinges on the following lemma.

	\begin{lem}
		\label{b_ralpha_half}
		There is a uniform family of $\omega$-regular structures $(\mathcal{C}_a)_{a\in R}$
		 such that for each $a\in R$ 
		\[
		\mathcal{C}_a \simeq ([-\alpha(a), \infty), <, +, \mathbb{N}, \alpha(a) \mathbb{N}).
		\]
	\end{lem}

\begin{proof}[Proof of Theorem \ref{b_ralpha_2half}]
Let $(\mathcal{C}_a)_{a\in R}$ be  an uniform family of $\omega$-regular structures as given by Lemma \ref{b_ralpha_half}. Within $\mathcal{C}_a$, define the set $L = \{x \in [-\alpha(a), \infty) : x \geq 0\}$, where $0$ is the $<$-least element of $\mathbb{N}$. This is an ordered commutative monoid. Let $L'$ be its Grothendieck group, and let $+', <'$ be the induced abelian group operation and ordering. There is a canonical inclusion map $\iota : L \hookrightarrow L'$. Let $Z' = \iota(\mathbb{N}) \cup -\iota(\mathbb{N})$ and $A' = \iota(\alpha(a) \mathbb{N}) \cup -\iota(\alpha(a) \mathbb{N})$. Observe that $(L', <', +', N', A')$ is an isomorphic copy of $\mathcal{R}_{\alpha(a)}$, defined in $\mathcal{C}_a$ in a manner uniform in $a$. So let $\mathcal{D}_a$ be this structure and conclude that $(\mathcal{D}_a)_a$ is a uniform family of $\omega$-regular structures.
\end{proof}

\noindent The proof of Lemma \ref{b_ralpha_half} itself is a uniform version of the argument given in \cite{H} that also fixes some minor errors of the original proof. By Lemma~\ref{lem:order-reg} and Theorem \ref{oplus_fin_regular}, we already know that 
\[
Z_v : (A_v^{\rm fin}, \prec_v^{\rm fin}, \oplus_v^{\rm fin}) \to (\mathbb{N}, <, +)
\]
is an isomorphism for every $v\in R$. As our eventual goal also requires us to define the set $\alpha \mathbb{N}$, it turns out to be much more natural to instead use the isomorphism 
\[
\alpha(v) Z_v: (A_v^{\rm fin}, \prec_v^{\rm fin}, \oplus_v^{\rm fin})\to (\alpha(v) \mathbb{N}, <, +)
\]
and recover $\N$ (and further $\Z$). We do so by following (and correcting) the argument in \cite{H}.

\begin{lem}\label{lem:ostcarry}
Let $v\in R$, and let $t_1,t_2,t_3 \in A_v$ be such that $t_1 \oplus_v t_2 = t_3$. Then
\begin{align*}
O_v(t_1) + O_v(t_2) = \begin{cases} O_v(t_3)+1 &\mbox{if } \mathbf{0}_v \prec_v t_1 \text{ and } t_3 \prec_v t_2; \\
O_v(t_3) - 1& \mbox{if } t_1 \prec_v \mathbf{0}_v \text{ and  } t_2 \prec_v t_3; \\
O_v(t_3)& \text{otherwise.}\\
\end{cases} 
\end{align*}
\end{lem}	
\begin{proof}
For ease of notation, let $\alpha = \alpha(v)$, and set $x_i = O_v(t_i)$ for $i = 1, 2, 3$. By definition of $\oplus_v$, we have that $x_1, x_2, x_3 \in I_{\alpha(v)}$ with $x_1 + x_2 \equiv x_3 \pmod{1}$. Note that $t_i \prec_v t_j$ if and only if $x_i < x_j$.\newline

\noindent We first consider the case that $0 < x_1$ and $x_3 < x_2$. Thus $ x_1 + x_2> 1-\alpha$. 
Note that
\[
-\alpha= 1-\alpha -1 < x_1 + x_2 - 1 < (1 - \alpha) + (1 - \alpha) - 1 = 1 - 2 \alpha < 1 - \alpha.
\]
Thus $x_1 + x_2 - 1\in I_{\alpha}$ and $x_3 = x_1 + x_2 - 1$.\newline

\noindent Now assume that $x_1 < 0$ and $x_2 < x_3$. Then $x_1 + x_2 <-\alpha$, and therefore
\[
1 -\alpha > x_1 + x_2 + 1 \geq (-\alpha) + (-\alpha) + 1 = (1 - \alpha) - \alpha > - \alpha.
\]
Thus $x_1 + x_2 + 1\in I_{\alpha}$ and hence $x_3 = x_1 + x_2 + 1$.\newline

\noindent Finally consider that $0, x_1$ are ordered the same way as $x_2, x_3$. Since $x_1+x_2 \equiv x_3 \pmod{1}$, we know that $|x_1 - 0|$ and $|x_3 - x_2|$ differ by an integer $k$. If $k > 0$, would imply that one of these differences is at least $1$, which is impossible within the interval $I_{\alpha}$. Therefore $x_1 - 0 = x_3 - x_2$ and hence $x_3 = x_1 + x_2$.
\qedhere 
\end{proof}

\noindent For $i\in \N$, set $
\mathbf{i}_{v} := \underbrace{\mathbf{1}_v \oplus \dots \oplus \mathbf{1}_v}_{i \hbox{ times }}.$

\begin{lem}\label{lem:f}
The set $F := \{ (v,s) \in A^{\rm fin} \ : \ Z_v(s)\alpha(v) < 1\}$
is $\omega$-regular, and for each $(v,s) \in F$
\[
O_v(s) =
\begin{cases}
\alpha(v) Z_v(s) & \hbox{if $(\alpha(v)+1) Z_v(s)<1$;}\\
\alpha(v) Z_v(s) -1 & \hbox{otherwise.}
\end{cases}
\]
\end{lem}

\begin{proof}
By Lemma \ref{less_than_half}, we can first consider the case that $\alpha(v) > \frac{1}{2}$. In this situation, $F_v$ is just the set $\{\mathbf{0}_v, \mathbf{1}_v\}$, and hence obviously $\omega$-regular.\newline

\noindent Now assume that $\alpha(v) < \frac{1}{2}$. Let $w$ be the $\prec^{\rm fin}_v$-minimal element of $A^{\rm fin}_v$ with $w \prec_v \mathbf{0}_v$.   We will show that 
\[
F_{v} = \{ s \in A^{\rm fin}_v \ : \ s \preceq_v^{\rm fin} w\}.
\]
Then $\omega$-regularity of $F$ follows then immediately.\newline 

\noindent 
Let $n\in \N$ be maximal such that $n\alpha(v)<1$. It is enough to show that $Z_v(w)=n$. 
By Lemma \ref{lem:Oofone}, $O_v(\mathbf{1}_v)=\alpha(v)$. Hence $1\alpha(v),2\alpha(v),\dots,(n-1)\alpha(v) \in I_{\alpha(v)}$, but $n\alpha(v) >1-\alpha(v)$.
Then for $i=1,\dots,n-1$
\[
O_v(\mathbf{i}_v) = i\alpha(v), \ O_v(\mathbf{n}_v) = n\alpha(v) - 1 < 0. 
\]
So $\mathbf{i}_v\succeq \mathbf{0}_v$ for $i=1,\dots, n$, but $\mathbf{n}_v \prec \mathbf{0}_v$. Thus $\mathbf{n}_v = w$ and $Z_v(w)=n$.
\end{proof}

\begin{lem}\label{lem:ostsmz}
Let $v\in R$ and $t\in A_v^{\rm fin}$. Then there is an $s \in F_v$ and $t'\in A_v^{\rm fin}$ such that $t' \preceq_v 0$ and $t = t'\oplus_v s$. In particular,
\[
A_v^{\rm fin} = \{ t \in A_v^{\rm fin} \ : \ t \preceq_v \mathbf{0}_v\} \oplus_v F_v.
\]
\end{lem}

\begin{proof}
Let $n\in \N$ be maximal such that $n\alpha(v) < 1$. 
Let $t\in A_v^{\rm fin}$. We need to find $s\in A_v^{\rm fin}$ and $u\in F_v$ such that $t=s\oplus_v^{\rm fin} u$. We can easily reduce to the case that $t\succ \mathbf{0}_v$ and $Z_v(t)>n$.\newline

\noindent Let $i\in \{0,\dots,n\}$ be such that $0\geq O_v(t)-i\alpha(v) >  - \alpha(v)$. Then let $s\in A_v^{\rm fin}$ be such that $Z_v(s)=Z_v(t)-i$. Note $t = s \oplus_v^{\rm fin} \mathbf{i}_v$. Thus we only need to show that $s\preceq \mathbf{0}_v$.\newline

\noindent To see this, observe that by Lemma \ref{lem:f}
\[
O_v(s) + \alpha(v) i\equiv O_v(s) + O_v(\mathbf{i}_v) \equiv  O_v(t) \pmod{1}. 
\]
Since $O_v(t)-i\alpha(v)\in I_{\alpha(v)}$, we know that  $O_v(s)=O_v(t)-i\alpha(v) \leq 0$. \\
Therefore $O_v(s)\preceq \mathbf{0}_v$.
\end{proof}

\begin{proof}[Proof of Lemma \ref{b_ralpha_half}]
Define $B\subseteq A^{\rm fin}$ to be 
$
\{(v,s)\in A^{\rm fin} \ : \ s \preceq_v \textbf{0}_v\}$.
Clearly, $B$ is $\omega$-regular. We now define $\prec^B$ and $\oplus^B$ such that for each $v\in R$, 
the structure $(B_v, \prec^B_v, \oplus^B_v)$ is isomorphic to $(\mathbb{N}, <, +)$ under the map $g_v$ defined as $g_v(s) = \alpha(v) Z_v(s) - O_v(s)$.\newline

\noindent We define $\prec^B$ to be the restriction of $\prec^{\rm fin}$ to $B$. That is, for $(v,s_1),(v,s_2)\in B$ we have
\[ 
(v,s_1) \prec^B (v,s_2) \hbox{ if and only if } (v,s_1) \prec^{\rm fin} (v,s_2).
\]
 It is immediate that $\prec^B$ is $\omega$-regular, since both $B$ and $\prec^{\rm fin}$ are $\omega$-regular.\newline

\noindent We define $\oplus^B$ as follows:

\begin{equation*}
(v,s_1) \oplus^B (v,s_2) = \begin{cases} (v,s_1\oplus_v s_2) &\mbox{if } s_1\oplus_v^{\rm fin} s_2 \preceq_v \mathbf{0}_v ;\\
(v,s_1\oplus_v s_2 \oplus_v \mathbf{1}_v) & \mbox{otherwise.} \end{cases}    
\end{equation*}

\noindent We now show that $g_v(s_1\oplus_v^B s_2)=g_v(s_1) + g_v(s_2)$ for every $s_1,s_2\in B_v$. \newline

\begin{table}[t]
    \centering
\begin{tabular}{ c|c }
Name & Definition\\
 \hline
 $A$ & $\{ (v,w) \ : \ v\in R, \hbox{ $w$ is a $\#$-$v$-Ostrowski representation}\}$\\
 $A^{\rm fin}$ & $\{ (v,w) \ : \ v\in R, \hbox{ $w$ is a $\#$-$v$-Ostrowski representation and eventually $0$}\}$\\
 $B$ & $\{(v,s)\in A^{\rm fin} \ : \ s \preceq_v \textbf{0}_v\}$ \\
 $C$ & $\{ (v,s,t) \ : \ (v,s) \in B \wedge (v,t) \in A \}$
\end{tabular}
    \caption{Definitions of sets used in the proof}
    \label{fig:my_label}
\end{table}

\noindent Let $(v,s_1),(v,s_2)\in B$. We first consider the case that
$s_1\oplus_v s_2 \preceq_v \mathbf{0}_v$. By Lemma \ref{lem:ostcarry}, $O_v(s_1\oplus_v s_2)= O_v(s_1) + O_v(s_2)$. Thus 
\begin{align*}
g_v(s_1 \oplus_v^B s_2)  &= g_v(s_1\oplus_v s_2)  \\
&= \alpha(v) Z_v(s_1\oplus_v s_2) - O_v(s_1\oplus_v s_2)\\
&=\alpha Z_v(s_1) + \alpha Z_v(s_2) - O_v(s_1) - O_v(s_2)\\
&=g_v(s_1) + g_v(s_2).
\end{align*}
Now suppose that	$s_1\oplus_v s_2 \succ_v \mathbf{0}_v$. Since $-\alpha(v)\leq O_v(s_1), O_v(s_2)\leq 0$, we get that
\[
1-\alpha(v) >  O_v(s_1) + O_v(s_2) +\alpha(v) \geq -\alpha(v). 
\]
Thus by Lemma \ref{lem:Oofone},
\[
O_v(s_1\oplus_v s_2\oplus_v \mathbf{1}_v) = O_v(s_1) + O_v(s_2) +\alpha(v).
\]
We obtain
\begin{align*}
g_v(s_1 \oplus_v^B s_2)  &= g_v(s_1\oplus_v s_2 \oplus_v \mathbf{1}_v) \\
&= \alpha Z_v(s_1\oplus_v s_2\oplus_v \mathbf{1}_v) - O_v(s_1\oplus_v s_2\oplus_v \mathbf{1}_v)\\
&=\alpha(v)\big(Z_v(s_1) + Z_v(s_2)\big) + \alpha(v) - O_v(s_1) - O_v(s_2) - \alpha(v)\\
&=g_v(s_1) + g_v(s_2).
\end{align*}
Since $s_1\prec_v s_2$ if and only if $Z_v(s_1)<Z_v(s_2)$, we get that $g_v$ is an isomorphism between $(B_v, \prec^B_v, \oplus^B_v)$ and $(\mathbb{N}, <, +)$.\newline

\begin{table}[t]
    \centering
\begin{tabular}{ c|cc }
Map & Domain & Codomain\\
 \hline
 $\alpha$ & $R$ & $\Irr$\\
 $O_v$ & $A_v$ & $I_{\alpha(v)}$\\
 $Z_v$ & $A_v^{fin}$ & $\N$\\
 $g_v:=\alpha(v)Z_v - O_v$ & $B_v$ & $\N$ \\
 $T_v:=g_v + O_v$ & $C_v$ & $[-\alpha(v),\infty)\subseteq \R$
\end{tabular}
    \caption{A list of the maps and their domains and codomains.}
    \label{fig:my_label2}
\end{table}

\noindent Let $C$ be defined by 
\[
\{ (v,s,t) \in (\SigmaH^{\omega})^3\ : \ (v,s) \in B \wedge (v,t) \in A \}.
\]
 Clearly $C$ is $\omega$-regular. Let $T_v : C_v \to [-\alpha(v), \infty) \subseteq \mathbb{R}$ map $(s,t) \mapsto g_v(s) + O_v(t)$.\newline

\noindent Note that $T_v$ is bijective for each $v\in R$, since every real number decomposes uniquely into a sum $n+y$, where $n\in \Z$ and $y\in I_v$. \newline

\noindent We define an ordering $\prec^C_v$ on $C_v$ lexicographically: $(s_1, t_1) \prec^C_v (s_2,t_2)$ if either
\begin{itemize}
    \item $s_1 \prec^B_v s_2$, or
    \item $s_1 = s_2$ and $t_1 \prec_v t_2$.
\end{itemize}
The set
\[
\{(v,s_1,t_1,s_2,t_2) \ : \ (s_1,t_1), (s_2, t_2)\in C_v \wedge (s_1,t_1) \prec_v^C (s_2, t_2)   \}
\]
is $\omega$-regular. We can easily check that 
$(s_1, t_1) \prec^C_v (s_2,t_2)$ if and only if $ T_v(s_1,t_1) < T_v(s_2,t_2)$.

\vspace{0.2cm}

\noindent Let $\mathbf{0}^B$ be $g_v^{-1}(0)$ and $\mathbf{1}^B$ be $g_v^{-1}(1)$. Let $\ominus^B$ be the (partial) inverse of $\oplus^B$. We define $\oplus^C$ for $(s_1,t_1),(s_2,t_2)\in C$ as follows:
\[
(s_1,t_1) \oplus^C_v (s_2,t_2) = \begin{cases} (s_1 \oplus_v^B s_2 \ominus^B \mathbf{1}^B, t_1 \oplus_v t_2) &\mbox{if } t_1 \prec \mathbf{0}_v \wedge t_2 \prec_v t_1 \oplus_v t_2; \\
(s_1 \oplus_v^B s_2 \oplus_v^B \mathbf{1}^B, t_1 \oplus_v t_2) & \mbox{if } \mathbf{0}_v \prec t_1 \wedge t_1 \oplus_v t_2 \prec_v t_2;\\
(s_1 \oplus^B_v s_2, t_1 \oplus_v t_2)& \text{otherwise.}\\
\end{cases}
\]
(Note that $\oplus^C$ is only a partial function, as the case where $s_1 = s_2 = \mathbf{0}^B$ and $t_1 \prec \mathbf{0}_v \wedge t_2 \prec_v t_1 \oplus_v t_2$ is outside of the domain of $\ominus^B$.) It is easy to check that $\oplus^C$ is $\omega$-regular. It follows directly from Lemma \ref{lem:ostcarry} that
\[
T_v((s_1,t_1) \oplus^C_v (s_2,t_2)) = T_v((s_1,t_1)) + T_v((s_2,t_2)).
\]
Thus for each $v \in R$, the function $T_v$ is an isomorphism between $(C_v,\prec^C_v,\oplus^C_v)$ and $([-\alpha(v),\infty),<,+)$. \noindent To finish the proof, it is left to establish the $\omega$-regularity of the following two sets:
\begin{enumerate}
    \item  $\{ (v,s,t) \in C \ : \ T_v(s,t) \in  \N\}$,
\item    $\{ (v,s,t) \in C \ : \ T_v(s,t) \in \alpha(u) \N\}$.
\end{enumerate}

\vspace{0.2cm}

\noindent For (1), observe that the set $T_v^{-1}(\mathbb{N})$ is just the set $\{(s,t) \in C_v : t = \mathbf{0}_v\}$.\newline  

\noindent For (2), consider the following two sets:
\begin{itemize}
    \item $U_1 = \{ (v,s,t)\in C \ : \ s=t\}$,
    \item $U_2 = \{ (v,\mathbf{0}_v,t) \in C \ : \ t \in F_v\}$.
\end{itemize}
Let $\mathbf{1}^C_v$ be $T_v^{-1}(1)$. Set 
\begin{align*}
U := \{ (v,(s_1,t_1)&\oplus_v^c (\mathbf{0}_v,t_2)) \ : \ (v,s_1,t_1) \in U_1, (v,\mathbf{0}_v,t_2) \in U_2, t_2\succeq 0\}\\
\cup & \ \{ (v,(s_1,t_1)\oplus_v^c (\mathbf{0}_v,t_2)\oplus \mathbf{1}^C_v) \ : \ (v,s_1,t_1) \in U_1, (v,\mathbf{0}_v,t_2) \in U_2, t_2\prec 0\}.
\end{align*}
The set $U$ is clearly $\omega$-regular, since both $U_1$ and $U_2$ are $\omega$-regular. We now show that $T_v(U)=\alpha(v)\N.$\newline

\noindent Let $(v,s,s)\in U_1$ and $(v,\mathbf{0}_v,t)\in U_2$.
If $t\succeq \mathbf{0}_v$, then by Lemma \ref{lem:f}
\begin{align*}
T_v((s,s) \oplus_C (\mathbf{0}_v,t)) &=  T_v(s,s) + T_v(\mathbf{0}_v,t) \\
&= \alpha(v) Z_v(s) - O_v(s) + O_v(s) + O_v(t)\\
&= \alpha(v) Z_v(s) + \alpha(v) Z_v(t) = \alpha(v) Z_v(s\oplus_v t).
\end{align*}
If $t\prec \mathbf{0}_v$, then by Lemma \ref{lem:f}
\begin{align*}
T_v((s,s) \oplus^C_v (\mathbf{0}_v,t)\oplus^C_v \mathbf{1}_v^C) &=  T_v(s,s) + T_v(\mathbf{0}_v,t) + 1 \\
&= \alpha(v) Z_v(s) - O_v(s) + O_v(s) + O_v(t) + 1\\
&= \alpha(v) Z_v(s) + \alpha(v) Z_v(t) = \alpha(v) Z_v(s\oplus_v t).
\end{align*}
\noindent Thus $T_v(U) \subseteq \alpha(v)\N$. By Lemma \ref{lem:ostsmz}, $T_v(U)=\alpha(v)\N.$
\qedhere 
	\end{proof}

\section{Decidability results}
\label{section:decidability}

We are now ready to prove the results listed in the introduction. We first recall some notation. Let $\mathcal L_m$ be the signature of the first-order structure $(\R,<,+,\Z)$, and let $\mathcal L_{m,a}$ be the extension of $\mathcal L_m$ by a unary predicate. For $\alpha \in \R_{>0}$, let $\mathcal{R}_\alpha$ denote the $\mathcal L_{m,a}$-structure $(\mathbb{R}, <, +, \mathbb{Z}, \alpha \mathbb{Z})$.
For each $\mathcal L_{m,a}$-sentence $\varphi$, we set
\[
R_{\varphi} := \{ v \in R \ : \ \mathcal{R}_{\alpha(v)}\models \varphi \}.
\]
\begin{thm}\label{thm:rvarphiregular}
Let $\varphi$ be an $\mathcal L_{m,a}$-sentence. Then $R_{\varphi}$ is $\omega$-regular.
\end{thm}
\begin{proof}
By Theorem \ref{b_ralpha_2half} there is a uniform family of $\omega$-regular structures $(\mathcal{D}_v)_{v\in R}$ such that
		 such that 
		$
		\mathcal{D}_v \simeq \mathcal{R}_{\alpha(v)}$ for each $v\in R$.
Then $R_{\varphi} = \{ v \in R \ : \ \mathcal{D}_v \models \varphi\}.$ This set is $\omega$-regular by Fact \ref{fact:uniformregulardef}.
\end{proof}

\noindent Let $\mathcal N=(R;(R_{\varphi})_{\varphi},(X)_{X\subseteq R^n\ \omega\text{-regular}})$ be the relational structure on $R$
with the  relations $R_{\varphi}$ for every $\cal L$-sentences $\varphi$ and $X\subseteq R^n$ $\omega$-regular. Because $\cal N$ is an $\omega$-regular structure, we obtain the following decidability result.
\begin{cor}
The theory $\Th(\mathcal N)$ is decidable.       
\end{cor}

\noindent We now proceed towards the proof of Theorem C. Recall that $\Irr := (0,1)\setminus \Q.$

\begin{defi} Let $X\subseteq \Irr^n$. Let $X_R$ be defined by
\[
X_R :=\{ (v_1,\dots,v_n) \in R^n \ : \ v_1 \sim_{\#} v_2 \sim_{\#} \dots \sim_{\#} v_n \wedge (\alpha(v_1),\dots,\alpha(v_n)) \in X \}  
\]
We say $X$ is \textbf{recognizable modulo $\sim_{\#}$} if $X_R$ is $\omega$-regular.
\end{defi}

\begin{lem}\label{lem:closednessr}
The collection of sets recognizable modulo $\sim_{\#}$ is closed under Boolean operations and coordinate projections.
\end{lem}

	\begin{proof} 
Let $X, Y \subseteq \Irr$ be recognizable modulo $\sim_\#$. It is clear that $(X \cap Y)_R = X_R \cap Y_R$. Thus $X \cap Y$ is recognizable modulo $\sim_\#$. Let $X^c$ be  $\Irr^n\setminus X$, the complement of $X$. For ease of notation, set $
E := \{ (v_1,\dots,v_n) \in R^n \ : \ v_1 \sim_{\#} v_2 \sim_{\#} \dots \sim_{\#} v_n \}.$
Then
\begin{align*}
(X^c)_R &= \{ (v_1,\dots,v_n) \in R^n \ : \ v_1 \sim_{\#} v_2 \sim_{\#} \dots \sim_{\#} v_n \wedge (\alpha(v_1),\dots,\alpha(v_n)) \notin X \}\\
&= E\cap \{ (v_1,\dots,v_n) \in R^n \ :  (\alpha(v_1),\dots,\alpha(v_n)) \notin X \}\\
& = E \cap \{ (v_1,\dots,v_n) \in R^n \ :  (\alpha(v_1),\dots,\alpha(v_n)) \notin X \vee \neg (v_1 \sim_{\#} v_2 \sim_{\#} \dots \sim_{\#} v_n)\}\\
& = E \cap  (R^n\setminus X_R).
\end{align*}
This set is $\omega$-regular, and hence $X^c$ is recognizable modulo $\sim_\#$.\newline

\noindent For coordinate projections, it is enough to consider projections onto the first $n-1$ coordinates. Let $n > 0$ and let $\pi$ be the coordinate projection onto first $n-1$ coordinates. Observe that 
\[
\pi(X) = \{(\alpha_1, \dots, \alpha_{n-1}) \in \R^{n-1} \ : \ \exists \alpha_n \in \R \ (\alpha_1, \dots, \alpha_{n-1}, \alpha_n) \in X\}.
\]
Thus $\pi(X)_R$ is equal to
\[
\{ (v_1,\dots,v_{n-1}) \in R^{n-1} \ : \ v_1 \sim_{\#} \dots \sim_{\#} v_{n-1} \wedge 
\exists \alpha_n : (\alpha(v_1),\dots,\alpha(v_{n-1}), \alpha_n) \in X \}.
\]
Note that $v \mapsto \alpha(v)$ is a surjection $R \twoheadrightarrow (0, 1) \setminus \Q$. Thus $\pi(X)_R$ is also equal to:
\[\{ (v_1,\dots,v_{n-1}) \in R^{n-1} : v_1 \sim_{\#} \dots \sim_{\#} v_{n-1} \wedge 
\exists v_n : (\alpha(v_1),\dots,\alpha(v_n)) \in X \}.
\]
Unfortunately, this set is not necessarily equal to $\pi(X_R)$. There might be tuples\linebreak $(v_1,\dots,v_{n-1})$ such that no $v_n$ can be found, because it would require more bits in one of its coefficients than $v_1,\dots, v_{n-1}$ have for that coefficient. But $\pi(X_R)$ always contains \textit{some} representation of $\alpha(v_1),\dots, \alpha(v_{n-1})$ with the appropriate number of digits. We need only ensure that removal of trailing zeroes does not affect membership in the language. Thus $\pi(X)_R$ is just the zero-closure of $\pi(X_R)$. Thus $\pi(X)_R$ is $\omega$-regular by Lemma~\ref{zero_closure}.
\end{proof}

\begin{thm}\label{thm:dechastag}
Let $X_1,\dots, X_n$ be recognizable modulo $\sim_{\#}$ by B\"uchi automata $\cal A_1,\dots, \cal A_n$, and let $\cal Q$ be the structure $(\Irr; X_1,\dots, X_n)$. Then the theory of $\cal Q$ is decidable.
\end{thm}
\begin{proof}
By Lemma \ref{lem:closednessr} every set definable in $\cal Q$  is recognizable modulo $\sim_{\#}$. Moreover, for each definable set $Y$ the automaton that recognizes $Y$ modulo $\sim_{\#}$, can be computed from the automata $\cal A_1,\dots,\cal A_n.$
Let $\psi$ be a sentence in the signature of $\cal Q$. Without loss of generality, we can assume that $\psi$ is of the form $\exists x \ \chi(x).$ Set
\[
Z:= \{ a \in \Irr^n \ : \ \mathcal{Q} \models \chi(a) \}.
\]
Observe that $\mathcal{Q} \models \psi$ if and only if $Z$ is non-empty.
Note for every $a \in \Irr^n$ there are $v_1,\dots,v_n \in R$ such that $v_1 \sim_{\#} v_2 \sim_{\#} \dots \sim_{\#} v_n$ and  $(\alpha(v_1),\dots,\alpha(v_n))=a$. Thus $Z$ is non-empty if and only if 
\[
\{ (v_1,\dots,v_n) \in R^n \ : \ v_1 \sim_{\#} v_2 \sim_{\#} \dots \sim_{\#} v_n \wedge (\alpha(v_1),\dots,\alpha(v_n)) \in Z \}  
\]
is non-empty. Thus to decide whether $\cal Q \models \psi$, we first compute the automaton $\cal B$ that recognizes $Z$ modulo $\sim_{\#}$, and then check whether the automaton accepts any word. 
\qedhere 
\end{proof}
		
\noindent We are now ready to prove Theorem C; that is, decidability of the theory of the structure
\[
\mathcal{M} =(\Irr, <, (M_\varphi)_\varphi, (q)_{q \in \Qu}),
\]
where $M_{\varphi}$ is defined for each $\mathcal{L}_{m,a}$-formula as
\[
M_{\varphi} := \{ \alpha \in \Irr \ : \ \mathcal{R}_{\alpha}\models \varphi \}.
\]
	\begin{proof}[Proof of Theorem C]
We just need to check that the relations we are adding are all recognizable modulo $\sim_{\#}$. By Lemma \ref{continued_fraction_ordering_reg} the ordering $<$ is recognizable modulo $\sim_{\#}$. By Lemma \ref{continued_fraction_quadratic_reg}, the singleton $\{ q \}$ is is recognizable modulo $\sim_{\#}$ for every $q\in \Qu$. 
Since $M_{\varphi}=\alpha(R_{\varphi})$, recognizability  of $M_{\varphi}$ modulo $\sim_{\#}$ follows from Theorem \ref{thm:rvarphiregular}.	
	\end{proof}

\noindent We can add to $\cal M$ a predicate for every subset of $\Irr^n$ that is recognizable modulo $\sim_{\#}$, and preserve the decidability of the theory. The reader can check that examples of subsets of $\Irr$ recognizable modulo $\sim_{\#}$ are  the set of all $\alpha\in \Irr$ such that the terms in the continued fraction expansion of $\alpha$ are powers of 2, the set of all $\alpha\in \Irr$ such that the terms in the continued fraction expansion of $\alpha$ are in (or are not in) some fixed finite set, and  the set of all $\alpha\in \Irr$ such that all even (or odd) terms in their continued fraction expansion are 1.
	
\section{Automatically Proving Theorems about Sturmian Words}
We have created an automatic theorem-prover based on the ideas and the decision algorithms outlined above, called Pecan~\cite{pecan-repo}, available at 
\begin{center}
\url{https://github.com/ReedOei/Pecan}
\end{center}
We use Pecan to provide proofs of known and unknown results about characteristic Sturmian words. The Pecan code for the following examples is available at 
\begin{center}
\url{https://github.com/ReedOei/SturmianWords}
\end{center}
We quote some of this code throughout this section. These code snippets should be understandable without further explanation, but interested readers can find more information and explanations in \cite{Pecan}. We recommend downloading the code instead of copying from this paper. In addition to the size of the automata created by Pecan, we sometimes state the runtime of Pecan on a normal laptop to indicate how quickly these statements have been proved.
\label{section:pecan}

\subsection{Classical theorems}
We begin by giving automated proofs for several classical result result about Sturmian words. We refer the reader to \cite{zbMATH01737190} for more information and traditional proofs of these results.\newline

\noindent In the following, we assume that $a$ is irrational and $i,j,k,n,m,p,s$ are $a$-Ostrowski representations. This can be expressed in Pecan as
\begin{pecan}
Let a is bco_standard.
Let i,j,k,n,m,p,s are ostrowski(a).
\end{pecan}
Here \pecaninline{bco_standard} is a data type for real numbers encoded using $\#$-binary coding. Then \pecaninline{ostrowski(a)} determines the Ostrowski numeration system used for the variables \pecaninline{i,j,k,n,m,p} and {s}. Pecan allows the use of Unicode characters such as $\exists$, $\forall$, $\neg$ and $\wedge$, and we will use these here for readability. Of course, Pecan also supports writing \textcolor{red}{exists}, \textcolor{red}{forall}, \textcolor{red}{!} and \textcolor{red}{and} for the same operations.
We write $c_{a,0}(i)$ as \texttt{$\$$C[i]} in Pecan.\newline

\noindent Let $w^R$ denote the reversal of a word $w$. We say a word $w$ is  a \textbf{palindrome} if $w=w^R$.

\begin{thm}\label{thm:sturm_balanced_aper}
Characteristic Sturmian words are balanced and aperiodic.
\end{thm}

\begin{proof}
To show that a characteristic Sturmian word $c_{a,0}$ is balanced, it is sufficient to show that there is no palindrome $w$ in $c_{a,0}$ such that $0w0$ and $1w1$ are in $c_{a,0}$ (see ~\cite[Proposition 2.1.3]{zbMATH01737190}). We encode this in Pecan as follows.
The predicate \pecaninline{palindrome(a,i,n)} is true when $c_{a,0}[i..i+n] = c_{a,0}[i..i+n]^R$.
The predicate \pecaninline{factor_len(a,i,n,j)} is true when $c_{a,0}[i..i+n] = c_{a,0}[j..j+n]$.
Then Pecan takes \texttt{321.73} seconds to prove the following theorem:
\begin{pecan}
Theorem ("Balanced", {
foralla. !(existsi,n. palindrome(a,i,n) &
            (existsj. factor_len(a,i,n,j) & $\$$C[j-1] = 0 & $\$$C[j+n] = 0) &
            (existsk. factor_len(a,i,n,k) & $\$$C[k-1] = 1 & $\$$C[k+n] = 1))
}).
\end{pecan}

\noindent Encoding the property that a word is eventually periodic is straightforward:
\begin{pecan}
eventually_periodic(a, p) := 
    p > 0 & existsn. foralli. if i > n then $\$$C[i] = $\$$C[i+p]
\end{pecan}

\noindent The resulting automaton has $4941$ states and $35776$ edges, and takes \pecaninline{117.78} seconds to build.
We then state the theorem in Pecan, which confirms the theorem is true.
\begin{pecan}
Theorem ("Aperiodic", {
    foralla. forallp. if p > 0 then !eventually_periodic(a, p)
}). $\tikz[remember picture,baseline=(lastlstline.base)]{\node[anchor=base] (lastlstline) {};}$
\end{pecan}%
\tikz[remember picture,overlay]{
  \node[anchor=base east,inner sep=0pt] at (lastlstline.base -| current page text area.east) {\qedhere};
}%
\end{proof}

\noindent A word $w$ is a \textbf{factor} of a word $u$ if there exist words $v_1,v_2$ such that $u=v_1wv_2$. A factor $w$ of a word $u$ \textbf{right special} if both $w0$ and $w1$ are also factors of $u$.

\begin{thm}
    For each natural number $n$, $c_{a,0}$ contains a unique right special factor of length $n$, and this factor is $c_{a,0}[1..n+1]^R$.
\end{thm}
\begin{proof}
    We first define right special factors, as above.
    Recall that \pecaninline{factor_len(a,i,n,j)} checks that $c_{a,0}[i..i+n] = c_{a,0}[j..j + n]$.
\begin{pecan}
right_special_factor(a,i,n) := 
    (existsj. factor_len(a,i,n,j) & $\$$C[j+n] = 0) &
    (existsk. factor_len(a,i,n,k) & $\$$C[k+n] = 1)
\end{pecan}
\noindent We then define the \textbf{first} right special factor, which is the first occurrence (by index) of the right special factor in the word $c_{a,0}$.
    This step is purely to reduce the cost of checking the theorem: the \pecaninline{right_special_factor} automaton has $3375$ states, but \pecaninline{first_right_special_factor} has only $112$.
\begin{pecan}
first_right_special_factor(a,i,n) := special_factor(a,i,n) 
            & forallj. if (j>0 & factor_len(a,j,n,i)) then i<=j
\end{pecan}
We then check that each of these right special factors is equal to $c_{a,0}[1..n+1]^R$, which also proves the uniqueness.
The predicate \pecaninline{reverse_factor(a,i,j,l)} checks that $c_{a,0}[i..j] = c_{a,0}[k+1..l+1]^R$, where $j - i = l - k$. Then Pecan confirms:
\begin{pecan}
Theorem ("The unique special factor of length n is C[1..n+1]^R", {
foralla. foralli,n.
    if i > 0 & first_right_special_factor(a,i,n) then
        reverse_factor(a,i,i+n,n)
}). $\tikz[remember picture,baseline=(lastlstline.base)]{\node[anchor=base] (lastlstline) {};}$
\end{pecan}%
\tikz[remember picture,overlay]{
  \node[anchor=base east,inner sep=0pt] at (lastlstline.base -| current page text area.east) {\qedhere};
}%
\end{proof}

\noindent Another characterization of Sturmian words due to Droubay and Pirillo ~\cite[Theorem 5]{DROUBAY199973} is that a word is Sturmian if and only if it contains exactly one palindrome of length $n$ if $n$ is even, and exactly two palindromes of length $n$ if $n$ is odd.
We prove the forward direction below. 
\begin{thmC}[{\cite[Proposition 6]{DROUBAY199973}}]
    For every $n \in \N$, $c_{a,0}$ contains exactly one palindrome of length $n$ if $n$ is even, and exactly two palindromes of length $n$ if $n$ is odd.
\end{thmC}
\begin{proof}
    We begin by defining a predicate defining the location of the first occurrence of each length $n$ palindrome in $c_{a,0}$.
\begin{pecan}
first_palindrome(a, i, n) := palindrome(a, i, n) &
    forallj. if j > 0 & factor_len(a,j,n,i) then i <= j
\end{pecan}
    The resulting automaton has $247$ states and $1281$ edges.
    The following states the theorem, and Pecan proves it in \texttt{428.85} seconds.
\begin{pecan}
Theorem ("", {
foralla. foralln. ( 
if even(n) & n > 0 then
    existsi. forallk. first_palindrome(a,k,n) iff i = k ) & ( 
if odd(n) then
    existsi,j. i < j & forallk. first_palindrome(a,k,n) iff (i = k | j = k)
)}). $\tikz[remember picture,baseline=(lastlstline.base)]{\node[anchor=base] (lastlstline) {};}$
\end{pecan}%
\tikz[remember picture,overlay]{
  \node[anchor=base east,inner sep=0pt] at (lastlstline.base -| current page text area.east) {\qedhere};
}%
\end{proof}

\subsection{Powers}
Next, we prove the follow results about powers of Sturmian words. A finite nonempty subword $x$ of a (finite or $\omega$) word $w$ is a \textbf{$n$-th power} if $x = y^n$ for some finite word $y$.
We call a $2$nd power a \textbf{square}, and a $3$rd power a \textbf{cube}.\newline

\noindent 
Using Pecan, we construct an automaton recognizing the following property, stating that there is a square of length $n$ starting at $c_{a,0}(i)$:
\begin{pecan}
 $\texttt{square}$(a, i, n) := n > 0 & i > 0 
                 & forallj. i <= j & j < i + n & $\$$C[j]=$\$$C[j+n]. 
\end{pecan}
The resulting automaton has 80 states and 400 edges.
All characteristic Sturmian words contain such a square, as Pecan proves in \pecaninline{0.02} seconds:
\begin{pecan}
Theorem ("", {foralla. existsi,n. $\texttt{square}$(a, i, n)}).
\end{pecan}
\noindent Of course, it is easy to see all binary words of length at least four contain squares.
However, it is still useful to have created an automaton for recognizing squares, because it encodes quite a bit more information than just that squares exist: it also tells us exactly where they are in the Sturmian word. This allows Pecan to prove the following result.
 
\begin{thm}[{Dubickas \cite[Theorem 1]{Dubickas2009}}]
    All characteristic Sturmian words start with arbitrarily long squares.
\end{thm}
\begin{proof}
Using Pecan and the automaton for squares that we constructed earlier, we prove the following theorem, which takes \pecaninline{0.40} seconds.
\begin{pecan}
Theorem ("", { foralla. forallj,n.  existsm. m> n & $\texttt{square}$(a, j, m)
}). $\tikz[remember picture,baseline=(lastlstline.base)]{\node[anchor=base] (lastlstline) {};}$
\end{pecan}%
\tikz[remember picture,overlay]{
  \node[anchor=base east,inner sep=0pt] at (lastlstline.base -| current page text area.east) {\qedhere};
}%

\end{proof}
\noindent Furthermore, we can use an automaton recognizing squares to efficiently build automata recognizing higher-powers. 
Indeed, we ask Pecan to construct an automaton recognizing the following property that there is a cube of length $n$ starting at $c_{a,0}(i)$, as follows:
\begin{pecan}
cube(a, i, n) := $\texttt{square}$(a, i, n) & $\texttt{square}$(a, i + n, n)
\end{pecan}
We can ask Pecan to prove the well-known fact that characteristic Sturmian words contain cubes:
\begin{pecan}
Theorem ("", { foralla. existsi,n.  cube(a, i, m)}).
\end{pecan}
Pecan proves this in \pecaninline{0.25} seconds.\newline

\noindent Similar to squares, we have the following property for cubic prefixes.
\begin{thm}\label{thm:longcubes}
Let $a \in (0,1)$. Then $c_{a,0}$ starts with arbitrarily long cubes if and only if the continued fraction of $a$ is not eventually 1.
\end{thm}
\begin{proof}
First, we manually build an automaton recognizing $a$ such that the continued fraction of $a$ is not eventually one, called $\texttt{eventually\_one}$. Pecan proves the following in \pecaninline{2.37} seconds:
\begin{pecan}
Theorem ("", { 
 foralla. ((!eventually_one(a)) iff (forallm. existsn. n>m cube(a, 1, n)))
 }). $\tikz[remember picture,baseline=(lastlstline.base)]{\node[anchor=base] (lastlstline) {};}$
\end{pecan}%
\tikz[remember picture,overlay]{
  \node[anchor=base east,inner sep=0pt] at (lastlstline.base -| current page text area.east) {\qedhere};
}%
\end{proof}
\noindent The proof of Theorem \ref{thm:longcubes} highlights the ability of our decision algorithm, and hence of Pecan, to not only determine whether statements hold for all irrational numbers, but also whether a statement holds for all elements of a subset that is recognizable modulo $\sim_{\#}$. 
Indeed, we can use Pecan to show that if the continued fraction of $a$ is not eventually 1, then $c_{a,0}$ contains a fourth power. To do so, 
we  construct a predicate that holds whenever there is a fourth power of length $n$ starting at $c_{a,0}(i)$:
\begin{pecan}
fourth_pow(a, i, n) := $\texttt{square}$(a, i, n) & cube(a, i + n, n)
\end{pecan}
Finally, Pecan proves the following in \pecaninline{0.56} seconds.
\begin{pecan}
Theorem ("", {
foralla.  if !eventually_one(a) then existsi,n. forth_pow(a,i,n)
)}).
\end{pecan}

\noindent The converse is not true. Although it is easy to see without Pecan why, we can also ask Pecan for counterexamples using the following commands.

\begin{pecan}
Restrict i, n are ostrowski(a).
has_fourth_pow(a) := existsi,n. n > 0 & fourth_pow(a,i,n)
Example (ostrowskiFormat, { 
    bco_standard(a) & eventually_one(a) & has_fourth_pow(a) 
}).
\end{pecan}

\noindent Pecan responds with:
\begin{pecan}
[(a,[6][3]([1])^$\omega$)]
\end{pecan}
This means that $a = [0,6,3,\overline{1}]$ is a counterexample. Recall that $a \in (0,1)$, so the first digit of the continued fraction is always $0$ and therefore omitted by Pecan. For this choice of $a$, the characteristic Sturmian word $c_{a,0}$ starts with $000001$. Thus there is a fourth power immediately at the beginning of $c_{a,0}$.

\subsection{Antisquares and more}
Let $w\in \{0,1\}^*$. We let  $\overline{w}$ denote the $\{0,1\}$-word obtained by replacing each $1$ in $w$ by $0$ and each $0$ in $w$ by $1$. A word $w\in \{0,1\}^*$ is an \textbf{antisquare} if $w = v\overline{v}$ for some $v\in \{0,1\}^*$.  
    We define $A_O : (0,1) \setminus \Q \to \N\cup \{\infty\}$ to map an irrational $a$ to the maximum order of an antisquare in $c_{a,0}$ if such a maximum exists, and to $\infty$ otherwise. We let $A_L : (0,1)  \setminus \Q\to \N \cup \{\infty\}$ map $a$ to the maximum length of an antisquare in $c_{a,0}$ if such a maximum exists and $\infty$ otherwise. Note that $A_L(a) = 2A_O(a)$.\newline
    
\noindent Recall that $w^R$ denotes the reversal of a word $w$. A word $w\in \{0,1\}^*$ is an \textbf{antipalindrome} if $w = \overline{w^R}$. We set $A_P : (0,1) \setminus \Q \to \N\cup \{\infty\}$ to be the map that takes an irrational $a$ to the maximum length of an antipalindrome in $c_{a,0}$ if such a maximum, and to $\infty$ otherwise.
    We will use Pecan to prove that $A_O(a),A_L(a)$ and $A_P(a)$ are finite for every $a$. While the quantities $A_O(a)$, $A_P(a)$ and $A_L(a)$ can be arbitrarily large, we prove the new results that the length of the Ostrowski representations of these quantities is bounded, independent of $a$.   \newline
    
    \noindent  Let $a\in (0,1)$ be irrational and $N\in \N.$ Let $|N|_{a}$ denote the length of the $a$-Ostrowski representation of $N$, that is the index of the last nonzero digit of $a$-Ostrowski representation of $N$, or 0 otherwise.

\begin{thm}\label{thm:alpha_a_finite}
For every irrational $a \in (0,1)$
    \begin{enumerate}[align=left]
        \item[(i)] $|A_O(a)|_{a} \leq 4$,
        \item[(ii)] $|A_P(a)|_{a} \leq 4$,
        \item[(iii)] $|A_L(a)|_{a} \leq 6$,
        \item[(iv)]\label{thm:bound-antisq-antipalin} $A_O(a) \leq A_P(a) \leq A_L(a) = 2A_O(a)$.
    \end{enumerate}
    There are irrational numbers $a, \beta\in (0,1)$ such that $A_O(a) = A_P(a)$ and $A_P(\beta) = A_L(\beta)$.
\end{thm}
\begin{proof}
    Using Pecan, we create automata which compute $A_O$, $A_P$, and $A_L$:
    \begin{align*}
        A_O(a, n) &:= \texttt{has\_antisquare}(a, n) \land \forall m. \texttt{has\_antisquare}(a, m) \implies m \leq n \\
        A_P(a, n) &:= \texttt{has\_antipalindrome}(a, n) \land \forall m. \texttt{has\_antipalindrome}(a, m) \implies m \leq n \\
        A_L(a, n) &:= \texttt{has\_antisquare\_len}(a, n) \land \forall m. \texttt{has\_antisquare\_len}(a, m) \implies m \leq n
    \end{align*}
       We build automata recognizing $a$-Ostrowski representations of at most $4$ and $6$ nonzero digits, called $\texttt{has\_4\_digits}(n)$ and $\texttt{has\_6\_digits}(n)$.
    Then we use Pecan to prove all the parts of the theorem by checking the following statement.
\begin{pecan}
Theorem ("(i), (ii), (iii), and (iv)", {
foralla. has_4_digits($\texttt{max\_antisquare}$(a)) &
    has_4_digits(max_antipalindrome(a)) &
    has_6_digits($\texttt{max\_antisquare\_len}$(a)) &
    $\texttt{max\_antisquare}$(a) <= max_antipalindrome(a) &
    max_antipalindrome(a) <= $\texttt{max\_antisquare\_len}$(a)
}).
\end{pecan}
    We also use Pecan to find examples of the equality: when $a = [0; 3,3,\overline{1}]$, we have $A_O(a) = A_P(a) = 2$, and when $a = [0; 4,2,\overline{1}]$, we have $A_P(a) = A_L(a) = 2$.
\end{proof}

\begin{thm}\label{thm:antisquares}
    For every irrational $a \in (0,1)$, all antisquares and antipalindromes in $c_{a,0}$ are either of the form $(01)^*$ or of the form $(10)^*$.
\end{thm}

\begin{proof}
    We begin by creating a predicate called \texttt{is\_all\_01} stating that a subword $c_{a,0}[i..i+n]$ is of the form $(01)^*$ or $(10)^*$.
    We do this simply stating that $c_{a,0}[k] \neq c_{a,0}[k+1]$ for all $k$ with $i \leq k < i + n - 1$.
\begin{pecan}
$\texttt{is\_all\_01}$(a,i,n) := 
    forallk. if i <= k & k < i+n-1 then $\$$C[k] !=  $\$$C[k+1]
\end{pecan}

    We can now directly state both parts of the theorem; Pecan proves both in \pecaninline{76.1} seconds.
    
\begin{pecan}
Theorem ("All antisquares are of the form (01)^* or (10)^*", {
foralla. foralli,n. if $\texttt{antisquare}$(a,i,n) then $\texttt{is\_all\_01}$(a,i,n)
}).

Theorem ("All antipalindromes are of the form (01)^* or (10)^*", {
foralla. foralli,n. if $\texttt{antipalindrome}$(a,i,n) then $\texttt{is\_all\_01}$(a,i,n)
}). $\tikz[remember picture,baseline=(lastlstline.base)]{\node[anchor=base] (lastlstline) {};}$
\end{pecan}%
\tikz[remember picture,overlay]{
  \node[anchor=base east,inner sep=0pt] at (lastlstline.base -| current page text area.east) {\qedhere};
}%
\end{proof}

\subsection{Least periods of factors of Sturmian words}
We now use Pecan to give short automatic proofs a result about the least period of factors of characteristic Sturmian words.\newline 

\noindent     The \emph{semiconvergents} $p_{n,\ell}$ and $q_{n,\ell}$ of a continued fraction $[0; a_1, a_2, \ldots]$ are defined so that
    \[
        \frac{p_{n,\ell}}{q_{n,\ell}} = \frac{\ell p_{n-1} + p_{n-2}}{\ell q_{n-2} + q_{n-2}}
    \]
    for $1 \leq \ell < a_n$.

\begin{thm}\label{thm:least-period-semiconvergent}
    Let $p$ be the least period of a factor of $c_{a,0}$. Then $p$ is the denominator of a semiconvergent of $a$; that is $p = q_{n,\ell}$ for some $n$ and $\ell$.
\end{thm}

\begin{proof}
We define when a number $p$ is a least period of a factor of $c_{a,0}$ as an automaton \pecaninline{lp\_occurs}, as follows:
\begin{pecan}
least_period(a,p,i,j) := p = min{n : period(a,n,i,j)}
lp_occurs(a,p) := existsi,j. i>0 & j>0 & least_period(a,p,i,j)
\end{pecan}
It is easy to recognize $a$-Ostrowski representations of denominators of semiconvergents of $a$, because they are simply valid representations of the form $[0\cdots01b]_a$, where $b$ is some valid digit.

\begin{pecan}
Theorem ("", 
    { foralla,p. if lp_occurs(a,p) then semiconvergent_denom(p) }).
\end{pecan}
Pecan proves the theorem in \texttt{5016.77} seconds.
\end{proof}

\noindent A word $w$ is called \textbf{unbordered} if the least period of $w$ is $|w|$. We now are ready to reprove Lemma 8 in Currie and Saari~\cite{Currie2009LeastPO}. This is originally due to de Luca and De Luca \cite{deluca}.
\begin{thm}\label{thm:least_period_unbordered}
    The least period of $c_{a,0}[i..j]$ is the length of the longest unbordered factor of $c_{a,0}[i..j]$.
\end{thm}

\begin{proof}
We have previously defined least periods, so we can easily define unbordered factors. Similarly, it is straightforward to define the longest unbordered subwords of $c_{a,0}$:
\begin{pecan}
max_unbordered_subfactor_len(a,i,j,n) := 
    n = max{m : existsk. i<=k & k+n<=j & least_period(a,n,k,k+n)}
\end{pecan}
\noindent Then the theorem we wish to prove is
\begin{pecan}
Theorem ("",{ foralla,i,j,p. if i>0 & j>i & p>0 then
   least_period(a,p,i,j) iff max_unbordered_subfactor_len(a,i,j,p)
}).
\end{pecan}
\noindent Pecan confirms the theorem is true.
\end{proof}

\subsection{Periods of the length-\texorpdfstring{$n$}{n} prefix} 
In \cite{gabric2020inequality} Gabric, Rampersand and Shallit characterize all periods of the length-$n$ prefix of a characteristic Sturmian word in terms of the lazy Ostrowski representation. We are able implement their argument in Pecan.\newline

\noindent Let $a$ be a real number with continued fraction expansion $[a_0; a_1, a_2, \ldots]$ and convergents $p_k/q_k \in \Q$.
We recall the definition of the \emph{lazy $a$-Ostrowski numeration system}~\cite{epifanio2012sturmian}.
\begin{fact}
    Let $X \in \N$.
    The \emph{lazy $a$-Ostrowski representation of $X$} is the unique word $b_N \cdots b_1$ such that
    \[
        X = \sum_{n = 0}^N b_{n+1} q_n
    \]
    where 
    \begin{enumerate}
        \item $0 \leq b_1 < a_1$;
        \item $0 \leq b_i \leq a_i$ for $i > 1$;
        \item if $b_i = 0$ then $b_{i-1} = a_i$ for all $i > 2$;
        \item if $b_2 = 0$, then $b_1 = a_1 - 1$;
    \end{enumerate}
\end{fact}

\begin{thmC}[{\cite[Theorem 6]{gabric2020inequality}}]\label{thm:periods_lazy_ostrowski}
Let $a$ be an irrational real number, and define $Y_n$ to be the length $n$ prefix of $c_{a,0}$.
Define $\text{PER}(n)$ to be the set of all periods of $Y_n$.
Then
\begin{enumerate}
    \item The number of periods of $Y_n$ is equal to the sum of the digits in the lazy Ostrowski representation of $n$.
    \item\label{num-period-An} Let the lazy Ostrowski representation of $n$ be $b_1 \cdots b_N$, and define
    \[
        A(n) = \left\{ i q_j + \sum_{j < k < N} b_{k+1} q_k :  1 \leq i \leq b_{j+1}~\text{and}~ 0 \leq j \leq N \right\}
    \]
    Then $\text{PER}(n) = A(n)$.
\end{enumerate}
\end{thmC}

\begin{proof}
    As in~\cite{gabric2020inequality}, we note that it is sufficient to prove only (\ref{num-period-An}).
    We begin by defining the sets, indexed by the slope \pecaninline{a}.
    The set of periods of subwords of $c_{a,0}$ can be defined by the formula $p > 0 \land c_{a,0}[i..j-p] = c_{a,0}[i + p..j]$, allowing us to create an automaton recognizing this set, which we call \pecaninline{period(a,p,i,j)}.
    This automaton is more expressive what what we need for this theorem, so we then simply take the periods of the prefixes of $c_{a,0}$, as follows:
\begin{pecan}
p is $\$$Per(a,n) := existss. s = 1 & period(a,p,s,n+1)
\end{pecan}
To define $A(n)$, we first define several auxiliary automata and notions. Earlier, we defined addition automata for the (greedy) Ostrowski numeration system, but we can also easily handle the lazy Ostrowski numeration system using an automaton recognizing
    \[
        \left\{ (a,x,y) : x,y \in A_a^{\rm fin}, x = \#x_1\#x_2\#\cdots, y = \#y_1\#y_2\#\cdots, \sum_{i=0}^{\infty} x_{i+1} q_i = \sum_{i=0}^{\infty} y_{i+1} q_i \right\}
    \]
    which we call \pecaninline{ost_equiv(a,x,y)}.
    The \pecaninline{lazy_ostrowski(a,n)} automaton checks whether \pecaninline{n} is a valid lazy \pecaninline{a}-Ostrowski representation.
    These automata allow us to convert between the two systems.
    
    To define $A(n)$, we break it up into smaller pieces; first, we wish to recognize the set 
    \[
        B(n) = \{ iq_j : 1 \leq j \leq b_{j+1}~\text{and}~0 \leq j \leq N \}.
    \]
    For each $x \in (\#(0|1)^*)^\omega$, denote by $|x|_{\rm fin}$ the length of the longest prefix $y$ of $x$ such that $x = yz$ where $z \in (\#0^*)^\omega$, or $\infty$ if there is no such prefix.
    We then create the following automata:
    \begin{itemize}
        \item \pecaninline{as_long_as(x,y)} recognizing the set $\{ (x,y) : |x|_{\rm fin} \geq |y|_{\rm fin} \}$.
    
        \item \pecaninline{has_1_digit(x)} recognizing the set $(\#0^*)^* ((\#0^*)|(\#(0|1)^* 1 (0|1)^*)) (\#0^*)^\omega$, i.e., words of the form $\#w_1\#w_2\#\cdots$ such that there is at most one $w_i$ such that $w_i \not\in 0^*$.
        
        \item \pecaninline{bounded_by(x,y)} recognizing the set 
        \[
            \{ (x,y) : \text{$x$ and $y$ are aligned}, x = \#x_1\#x_2\#\cdots, y = \#y_1\#y_2\#\cdots, \forall i. x_i \leq_{\rm lex} y_i\}
        \]
    \end{itemize}
    Then we can recognize the set $B(n)$ from above by 
\begin{pecan}
i > 0 & has_1_digits(i) & as_long_as(n,i) & bounded_by(i,n_l)
\end{pecan}
    where \pecaninline{n_l} is the lazy \pecaninline{a}-Ostrowski representation of \pecaninline{n}.
    
    The last automaton we need to create is \pecaninline{suffix_after(x,y,s)}, recognizing the set $\{ (x,y,s) : s = 0^{|x|_{\rm fin}} \cdot y[|x|_{\rm fin}..] \}$.
    We need this to be able to recognize the set of \pecaninline{a}-Ostrowski representations
    \[
        \{ m : 0 \leq j \leq N, m_l = 0^j n_l[j..N], \text{$m_l$ is the lazy \pecaninline{a}-Ostrowski representation of $Z_a(m)$} \}
    \]
    where $n_l$ is the lazy \pecaninline{a}-Ostrowski representation of $Z_a(n)$.
    
    Finally, we can put everything together and define $A(n)$, again indexed by the slope $a$, as:
\begin{pecan}
p is $\$$A(a,n) :=
    existsn_l,m_l. lazy_ostrowski(a,n_l) & ost_equiv(a,n,n_l) &
    existsm. ost_equiv(a,m_l,m) &
    existsi. i > 0 & has_1_digit(i) & as_long_as(n,i) &
        bounded_by(i,n_l) & suffix_after(i,n_l,m_l) &
        i + m = p
\end{pecan}

Finally, we can state the theorem directly, which Pecan confirms is true.
\begin{pecan}
Theorem ("6 (b)", { foralla. forallp,n. p is $\$$Per(a,n) iff p is $\$$A(a,n) 
}). $\tikz[remember picture,baseline=(lastlstline.base)]{\node[anchor=base] (lastlstline) {};}$
\end{pecan}%
\tikz[remember picture,overlay]{
  \node[anchor=base east,inner sep=0pt] at (lastlstline.base -| current page text area.east) {\qedhere};
}%
\end{proof}

\section{Conclusion and Outlook}

\subsection{Scalar multiplication}
Recall that for $\alpha \in \R_{>0}$ we use $\cal{R}_{\alpha}$ to denote the $\cal{L}_{m,a}$-structure $(\R,<,+,\Z,a\Z)$.  Let $
\lambda_{\alpha} : \R \to \R$ be the function mapping $x$ to $\alpha x$, and let $\mathcal{S}_{\alpha}$ denote the structure $(\R,<,+,\Z,\lambda_\alpha)$. It is clear that every set definable in $\cal{R}_{\alpha}$ is also definable in $\cal{S}_{\alpha}$. The inverse is known to be true for some $\alpha$: 
    By Hieronymi \cite[Theorem D]{H-multiplication}, the function $\lambda_{\alpha}$ is definable in $\cal{R}_{\alpha}$ if $\alpha=\sqrt{d}$ for some $d\in \Q$, and thus in this situation every set definable in $\cal{S}_{\alpha}$ is also definable in $\cal{R}_{\alpha}$.

\begin{prop}\label{prop:existsa}
There is $\alpha \in \R$ such that $\cal{R}_{\alpha}$ does not define $\lambda_{\alpha}$.
\end{prop}
\begin{proof}
By \cite[Theorem A]{H-multiplication} the theory $\Th(\cal{S}_{\alpha})$ is undecidable when $\alpha$ is not quadratic. Thus is enough to find a non-quadratic $\alpha$ such that the theory $\Th(\cal{R}_{\alpha})$ is decidable. To do so, it suffices by Theorem 5.1 to find some $v\in R$ such that ${\rm FO}(\mathcal{D}_v)$ is decidable but $\alpha(v)$ is non-quadratic. 

\noindent Let $U$ be the set $\{i! : i \in \N\}$. Define $u=u_1u_2\cdots \in \{0,1\}^{\omega}$ such that $u_{i} = 1$ if $i\in U$, and $u_i = 0$ otherwise. 
Let $v=v_1v_2\cdots\in \SigmaH^{\omega}$ be such that
\[
v_i = \begin{cases}
\# & i = 1, \\
\# & i > 2 \text{ and } u_i=1,\\
1 & \, \text{otherwise.}
\end{cases}
\]
 That is, 
 \[
 v = \#1111\#11111111111111111\#1\cdots.
 \]
 By Elgot and Rabin \cite[Proof of Theorem 5]{ER}, the acceptance problem for $u$ is decidable. 
This implies that the acceptance problem for $v$ is decidable as well. Thus the theory ${\rm FO}(\mathcal{D}_v)$ is decidable by Fact \ref{fact:parameter}. However, the coefficients of the continued fraction expansion of $\alpha(v)$ are unbounded. Since quadratic numbers have periodic continued fractions, we conclude that $\alpha(v)$ is not quadratic.
\end{proof}

\noindent As argued in the proof above, it follows from Fact \ref{fact:parameter} that for every $\alpha\in \Irr$ the theory $\Th(\cal{R}_{\alpha})$ is decidable whenever there is $v
\in R$ such that the acceptance problem for $v$ is decidable and $\alpha(v)=\alpha$. We leave it as an open question whether this sufficient condition is also necessary. It would be interesting to know whether there are any natural non-quadratic numbers, like $e$ or $\pi$, for which this condition is satisfied.\newline

\noindent Recall that $\cal{L}_m$ is the signature of $\Th(\R,<,+)$ together with a unary precidate symbol $P$. Let $\cal{L}_{m,\lambda}$ be the extension of $\cal{L}_m$ by a unary functions symbol $\lambda.$ We consider $\cal{S}_\alpha$ now as an $\cal{L}_{m,\lambda}$-structure. Let $\cal{K}_{\lambda}$ be the class of $\cal{L}_{m,\lambda}$-structures $\{\cal{S}_{\alpha}\ : \ \alpha \in \Irr\}$. By Proposition \ref{prop:existsa} there is no hope of using Theorem B to deduce the decidability of the theory $\Th(\cal{K}_{\lambda})$. Indeed, we can show the following.

\begin{prop}
The theory $\Th(\cal{K}_{\lambda})$ is undecidable.
\end{prop}
\begin{proof}
Consider the $\cal{L}_{m,\lambda}$-sentence $\psi$ 
\[
\forall x_1  \forall x_2 \forall x_3 \ (\bigwedge_{i=1}^3 P(x_i) \wedge \bigvee_{i=1}^3 x_i \neq 0) \rightarrow (\lambda(\lambda(x_1)) + \lambda(x_2) + x_3 \neq 0).   
\]
Hence
\[
\cal{S}_{\alpha} \models \psi \text{ if and only if $\alpha$ is not quadratic.}
\]
Consider $U = (Q, \Sigma, \sigma_1, \delta, q_1, q_2)$ be the universal 1-tape Turing machine with 8 states and 4 symbols as defined by Neary and Woods \cite{NW}. By the proof of \cite[Theorem 7.1]{HNP}\footnote{In \cite{HNP} it is only stated that for every non-quadratic $\alpha$ we can find such an $\cal{L}_{m,\lambda}$-sentence $\varphi_x$. However, it is clear from the given construction that the sentence does not depend on the particular $\alpha$.}, given an input $x \in \Sigma^*$, there is an $\cal{L}_{m,\lambda}$-sentence $\varphi_x$ such that for every non-quadratic $\alpha$
\[
\cal{S}_{\alpha}\models \varphi_x\text{ if and only if $U$ halts on input $x$.}
\]
Combining this, we have that given an input $x \in \Sigma^*$
\[
\Th(\cal{K}_{\lambda})\models \psi \rightarrow \varphi_{x}\text{ if and only if $U$ halts on input $x$.}
\]
Thus $\Th(\cal{K}_{\lambda})$ is undecidable.
\end{proof}

\noindent Let $\cal{K}_{q}$ be the class of all $\cal{L}_{m,a}$-structures $\cal{R}_{\alpha}$ with $\alpha\in \Irr$ quadratic, and similarly, let $\cal{K}_{\lambda,q}$ be the class of all $\cal{L}_{m,\lambda}$-structures $\cal{S}_{\alpha}$ with $\alpha\in \Irr$ quadratic. We leave it as an open question whether the theories $\Th(\cal{K}_{q})$ and $\Th(\cal{K}_{\lambda,q})$ are decidable. It is unlikely that that decidability of the latter theory could be deduced from the decidability of the theory $\Th(\cal{K}_{q})$,  because the definition of multiplication by $\sqrt{d}$ in the proof of \cite[Theorem D]{H-multiplication} depends on $d$.

\subsection{Computational complexity} By \cite[Theorem D]{H}, the structure $\mathcal{R}_{\alpha}$ defines an isomorphic copy of the standard model of the monadic second-order theory of $(\N,+1)$ whenever $\alpha \in \Irr$. Hence there can not be a decision algorithm for $\Th(\mathcal{K})$ whose computational complexity is in general lower than the complexity of the decision algorithm presented here. See \cite{HNP} for more detailed results for $\Th(\mathcal{S}_{\alpha})$ when $\alpha$ is quadratic.
It would still be interesting to know whether improvements can be obtained for specific fragments of these theories.\newline

\noindent If we are only interested in deciding statements about Sturmian words, we only need decidability of the less expressive theories $\mathcal{K}_{\operatorname{sturmian}}$ and $\mathcal{K}_{\operatorname{char}}$. Here we know very little about lower bounds for the computational complexity of these decision problems. In particular, we do not even know whether an analogue of \cite[Theorem D]{H}, stating the definability of an isomorphic copy of the standard model of the weak monadic second-order theory of $(\N,+1)$, holds for $\mathcal{N}_{\alpha,\rho}$, when $\alpha\in \Irr$.\newline

\noindent Better results are likely obtainable when dropping the order relation. For $
\alpha \in \Irr$, consider $\mathcal{Z}_{\alpha}:=(\Z,+,0,1,f_{\alpha})$, where $f_{\alpha}: \Z \to \Z$ is the function mapping $x$ to $\lfloor \alpha x \rfloor$. Khani and Zarei \cite{KhaniZarei} and Khani, Valizadeh and Zarei \cite{KVZ} prove quantifier-elimination results for such structures that have the potential to produce more efficient decision algorithms (see also G\"{u}nayd\i n and \"{O}zsahakyan \cite{GO}). However, the usual order relation of $\Z$ is unlikely to be definable in such structures, and therefore this setting might not be particularly useful to decide statements about Sturmian words.

\bibliographystyle{alphaurl}
\bibliography{biblio}

\end{document}